\documentclass[10pt]{IEEEtran}

\usepackage{epstopdf}
\usepackage{graphicx}
\usepackage{stmaryrd}
\usepackage{amsmath,amsthm,amssymb}
\usepackage{multirow,url}
\usepackage{color,cite}
\usepackage{algorithm}
\usepackage{algorithmicx}
\usepackage{algpseudocode}
\usepackage{url}

\newtheorem{lem}{Lemma}

\newtheorem{proposition}{Proposition}

\newtheorem{thm}{Theorem}

\ifodd 0
\newcommand{\rev}[1]{{\color{blue}#1}} 
\newcommand{\com}[1]{\textbf{\color{red} (COMMENT: #1)}} 
\newcommand{\comg}[1]{\textbf{\color{green} (COMMENT: #1)}}
\newcommand{\response}[1]{\textbf{\color{magenta} (RESPONSE: #1)}} 
\else
\newcommand{\rev}[1]{#1}
\newcommand{\com}[1]{}
\newcommand{\comg}[1]{}
\newcommand{\response}[1]{}
\fi

\begin{document}

\title{Adaptive Channel Recommendation For Opportunistic Spectrum Access}

\author{Xu Chen$^{\ast}$, Jianwei Huang$^{\ast}$, Husheng Li$^{\dagger}$\\$^{\ast}$Department of Information Engineering, The Chinese University of Hong Kong, Hong Kong\\
$^{\dagger}$Department of Electrical Engineering and Computer Science, The University of Tennessee Knoxville, TN, USA\\
email:\{cx008,jwhuang\}@ie.cuhk.edu.hk,husheng@eecs.utk.edu
}

\maketitle


\begin{abstract}
We propose a dynamic spectrum access scheme where secondary users cooperatively recommend \textquotedblleft{}good\textquotedblright{} channels to each other and access accordingly. We formulate the problem as an average reward based
Markov decision process. We show the existence of the optimal
stationary spectrum access policy, and explore its structure
properties in two asymptotic cases. Since the action space of the  Markov decision process is continuous, it is difficult to find the optimal policy by simply discretizing the action space and use the policy iteration, value
iteration, or Q-learning methods.  Instead, we propose a new
algorithm based on the Model Reference Adaptive Search method, and prove its convergence to the optimal policy. Numerical
results show that the proposed algorithms achieve up to $18\%$ and $100\%$ performance improvement than the static channel recommendation scheme in homogeneous and heterogeneous channel environments, respectively, and is more robust to channel dynamics. 
\end{abstract} 

\section{Introduction}
Cognitive radio technology enables unlicensed secondary wireless users to opportunistically share the spectrum with licensed primary users, and thus offers a promising solution to address the spectrum under-utilization problem \cite{key-6}. Designing an efficient spectrum access mechanism for cognitive radio networks, however, is challenging for several reasons: (1) \emph{time-variation}: spectrum opportunities available for secondary users are often time-varying due to primary users' stochastic activities \cite{key-6}; and (2) \emph{limited observations}: each secondary user often has a limited view of the spectrum opportunities due to the limited spectrum sensing capability \cite{key-7}. Several characteristics of the wireless channels, on the other hand, turn out to be useful for designing efficient spectrum access mechanisms: (1) \emph{temporal correlations}: spectrum availabilities are correlated in time, and thus observations in the past can be useful in the near future \cite{key-19}; and (2) \emph{spatial correlation}: secondary users close to one another may experience similar spectrum availabilities \cite{key-1}. In this paper, we shall explore the time and space correlations and propose a recommendation-based collaborative spectrum access algorithm, which achieves good communication performances for the secondary users.

Our algorithm design is directly inspired by the recommendation system in the electronic commerce industry. For example, existing owners of  various products can provide recommendations (reviews) on Amazon.com, so that other potential customers can pick the products that best suit their needs. Motivated by this, Li in \cite{key-2} proposed a static channel recommendation scheme, where secondary users recommend the channels they have
successfully accessed to nearby secondary users. Since each secondary user originally only has a limited view of spectrum availability, such information exchange enables secondary users to take advantages of the correlations in time and space, make more informed decisions, and achieve a high total transmission rate.

The recommendation scheme in \cite{key-2}, however, is rather static and does not dynamically change with network conditions. In particular, the static scheme ignores two important characteristics of cognitive radios. The first one is the \emph{time variability} we mentioned before. The second one is the \emph{congestion effect}. As depicted in Figure \ref{ChannelRec}, too many users accessing the same good channel leads to congestion and a reduced rate for everyone.

\begin{figure}[tt]
\begin{center}
\includegraphics[scale=0.5]{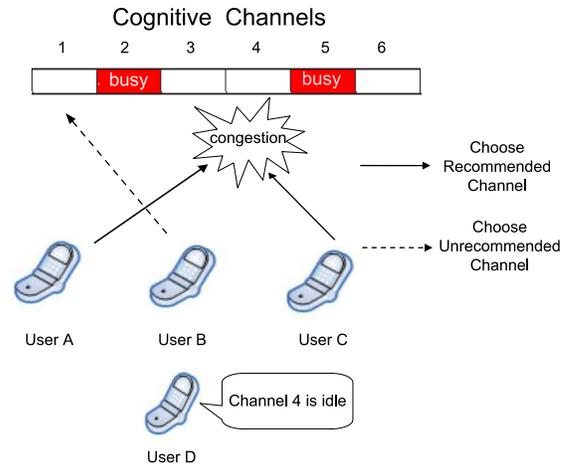}
\caption{\label{ChannelRec}Illustration of the channel recommendation scheme. User D recommends channel 4 to other users. As a result, both user A and user C access the same channel 4, and thus lead to congestion and a reduced rate for both users.}
\end{center}
\end{figure}

To address the shortcomings of the static recommendation scheme, in this paper we propose an adaptive channel recommendation
scheme, which adaptively changes the spectrum access probabilities based on users' latest channel recommendations. We formulate and analyze the system as a Markov decision process (MDP), and propose a numerical algorithm that always converges to the optimal spectrum access policy.

The main results and contributions of this paper  include:
\begin{itemize}
\item \emph{Markov decision process formulation}: we formulate and analyze the optimal recommendation-based spectrum access as an average reward MDP.
\item \emph{Existence and structure of the optimal policy}: we show that there always exists a stationary optimal spectrum access policy, which requires only the channel recommendation information of the most recent time slot. We also explicitly characterize the structure of the optimal stationary policy in two asymptotic cases (either the number of users or the number of users goes to infinity).
\item \emph{Novel algorithm for finding the optimal policy}: we propose
an algorithm based on the recently developed Model Reference Adaptive Search method \cite{key-5} to find the optimal stationary spectrum access policy.
The algorithm has a low complexity even when dealing with a continuous action space of the MDP. We also show that it always converges to the optimal stationary policy.
\item \rev{\emph{Superior Performance}: we show that the proposed algorithm achieves up to $18\%$ performance improvement than the static channel recommendation scheme  and $10\%$ performance improvement than the Q-learning method, and is also robust to channel dynamics.}

\end{itemize}

The rest of the paper is organized as follows. We introduce the system
model and the static channel recommendation scheme in Sections \ref{sec:System-Model} and \ref{sec:Static-Channel-Recommendation}, respectively. We then discuss the motivation for designing an adaptive channel recommendation scheme in Section \ref{sec:Motivation-Adaptive}. The Markov decision
process formulation and the structure results of the optimal policy
are presented in Section \ref{sec:Adaptive-Channel-Recommendation},
followed by the Model Reference Adaptive Search based algorithm in Section
\ref{sec:Model-Reference-Adaptive}. We illustrate the performance of the algorithm through numerical results in
Section \ref{sec:Numerical-Results}. We discuss the related
work in Section \ref{sec:Related-Work} and conclude in Section \ref{sec:Conclusion}. 

\section{System Model}\label{sec:System-Model}
\rev{
We consider a cognitive radio network with $M$ parallel and stochastically heterogeneous primary
channels. $N$ homogeneous secondary users try to access these channels using a slotted transmission structure (see Figure \ref{fig:Structure-of-Each}). The secondary users  can exchange information by broadcasting messages over a common control channel\footnote{Please refer to \cite{key-20} for the details on how to set up and maintain a reliable common control channel in cognitive radio networks.}. We assume that the secondary users are located close-by, thus they experience similar spectrum availabilities and can hear one another's broadcasting messages. To protect the primary transmissions, secondary users need to sense
the channel states before their data transmission.

The system model
is described as follows:
\begin{itemize}
\item \emph{Channel state:} For each primary channel $m$, the channel state at time slot $t$ is \[
S_{m}(t)=\begin{cases}
0, & \mbox{if channel $m$ is occupied by}\\
& \mbox{primary transmissions,}\\
1, & \mbox{if channel $m$ is idle.}
\end{cases}\]

\item \emph{Channel state transition:} The states of different channels change according to independent Markovian processes (see Figure
\ref{fig:Markovian-Channel-Model}). We denote the channel state probability
vector of channel $m$ at time $t$ as $
\boldsymbol{p}_{m}(t)\triangleq(Pr\{S_{m}(t)=0\},Pr\{S_{m}(t)=1\}),$
 which follows a two-state Markov chain as $
\boldsymbol{p}_{m}(t)=\boldsymbol{p}_{m}(t-1)\Gamma_{m},\forall t\geq 1,$
 with the transition matrix \[
\Gamma_{m}=\left[\begin{array}{cc}
1-p_{m} & p_{m}\\
q_{m} & 1-q_{m}\end{array}\right].\]
Note that when $p_{m}=0$ or $q_{m}=0$, the channel
state stays unchanged. In the rest of the paper, we will look at the more interesting and challenging cases where $0<p_{m}\leq1$ and $0<q_{m}\leq1$. The stationary distribution
of the Markov chain is given as \begin{eqnarray}
\lim_{t\rightarrow\infty}Pr\{S_{m}(t) & = & 0\}=\frac{q_{m}}{p_{m}+q_{m}},\label{eq:sd-1}\\
\lim_{t\rightarrow\infty}Pr\{S_{m}(t) & = & 1\}=\frac{p_{m}}{p_{m}+q_{m}}.\label{eq:sd-2}\end{eqnarray}
%

\item \emph{Heterogeneous channel throughput:} When a secondary user transmits successfully on an idle channel $m$, it achieves a data rate of $B_{m}$. Different channels can support different data rates.
\item \emph{Channel Contention:} To resolve the transmission collision when multiple secondary users access the same channel, a backoff mechanism is used (see Figure \ref{fig:Structure-of-Each} for illustration). The contention stage of a time slot is divided into $\lambda^{*}$ mini-slots, and each user $n$ executes the following two steps:
\begin{enumerate}
\item Count down according to a randomly and uniformly chosen integral backoff time (number of mini-slots) $\lambda_{n}$ between $1$ and $\lambda^{*}$.
\item Once the timer expires, monitor the channel and transmit RTS/CTS messages to grab the channel if the channel is clear (i.e., no ongoing transmission). Note that if multiple users choose the same backoff mini-slot, a collision will occur with RTS/CTS transmissions and no users can grab the channel. Once successfully grabing the channel, the user starts to transmit its data packet.
\end{enumerate}
Suppose that $k_{m}$ users choose channel $m$ to access. Then
the probability that user $n$ (out of the $k_{m}$
users) successfully grabs the channel $m$ is \begin{eqnarray}
Pr_{n} & = & Pr\{\min\{\lambda_{1},...,\lambda_{k_{m}}\}=\lambda_{n}\}\nonumber\\
 &  & \cdot \sum_{\lambda=1}^{\lambda^{*}}Pr\{\lambda_{n}=\lambda\}Pr\{\min_{i\neq n}\{\lambda_{i}\}>\lambda|\lambda_{n}=\lambda\}\nonumber\\
 & = & \frac{1}{k_{m}}\sum_{\lambda=1}^{\lambda^{*}}\frac{1}{\lambda^{*}}\left(\frac{\lambda^{*}-\lambda}{\lambda^{*}}\right)^{k_{m}-1}.\end{eqnarray}
For the ease of exposition, we focus on the asymptotic case where $\lambda^{*}$ goes to $\infty$. This is a good approximation when the number of mini-slots $\lambda^{*}$ for backoff is much larger than the number of users $N$ and collisions rarely occur. It simplifies the analysis as \begin{equation} \lim_{\lambda^{*}\rightarrow\infty}\frac{1}{\lambda^{*}}\sum_{\lambda=1}^{\lambda^{*}}(\frac{\lambda^{*}-\lambda}{\lambda^{*}})^{k_{m}-1}=1,\end{equation} and thus the expected throughput of user $n$ is \begin{equation}
u_{n}(t)=\frac{B_{m}S_{m}(t)}{k_{m}}.\label{eq:uu}\end{equation}
\end{itemize}
}

\begin{figure}[tt]
\begin{center}
\includegraphics[scale=0.6]{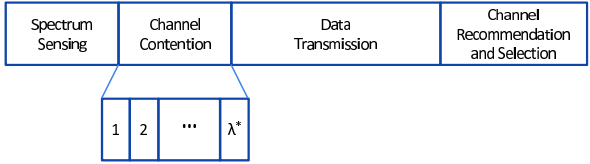}
\caption{\label{fig:Structure-of-Each}Structure of each spectrum access time
slot}
\end{center}
\end{figure}

\begin{figure}[tt]
\begin{center}
\includegraphics[scale=0.9]{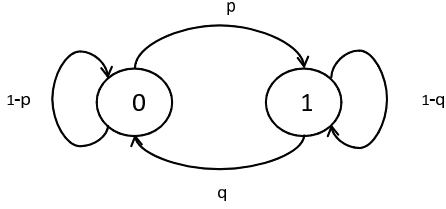}

\caption{\label{fig:Markovian-Channel-Model}Two states Markovian channel model}
\end{center}
\end{figure}


\section{\label{sec:Static-Channel-Recommendation1}Introduction To Channel Recommendation}
In this section, we first give a review of the static channel recommendation scheme in in  \cite{key-2} and then discuss the motivation for adaptive channel recommendation.
\subsection{\label{sec:Static-Channel-Recommendation}Review of Static Channel Recommendation}
The key idea of the static channel recommendation scheme is that secondary users inform each other about the available channels they have just accessed. More specifically, each secondary user executes the following four stages synchronously during each time slot (See Figure \ref{fig:Structure-of-Each}):
\begin{itemize}
\item \emph{Spectrum sensing:} sense one of the channels based on channel selection result made at the end of the previous time slot.
\item \emph{Channel Contention:} if the channel sensing result is idle, compete for the channel with the backoff mechanism described in Section \ref{sec:System-Model}.
\item \emph{Data transmission:}  transmit data packets if the user successfully grabs the channel.
\item \emph{Channel recommendation and selection:}
\begin{itemize}
\item \emph{Announce recommendation:} if the user has successfully accessed an idle channel, broadcast this channel ID to all other secondary users.
\item \emph{Collect recommendation:} collect recommendations
from other secondary users and store them in a buffer. Typically,
the correlation of channel availabilities between two slots diminishes as the time difference increases.
Therefore, each secondary user will only keep the recommendations received from the most recent $W$ slots and discard the out-of-date information. The user's own successful transmission history within $W$ recent time slots is also stored in the buffer. $W$ is a system design parameter and will be further discussed later.
\item \emph{Select channel}: choose a channel to sense at the next time slot by putting more weights on the recommended channels according to a \emph{static branching probability $P_{rec}$}. Suppose that the user has $0<R<M$ different channel recommendations in the buffer, then the probability of accessing
a channel $m$ is \begin{equation}
P_{m}=\begin{cases}
\frac{P_{rec}}{R}, & \mbox{if channel $m$ is recommended,}\\
\frac{1-P_{rec}}{M-R}, & \mbox{otherwise.}\end{cases}\label{eq:SCR-1}\end{equation}
 A larger value of $P_{rec}$ means that putting more weight on the recommended channels. When $R=0$ (no channel is recommended) or $M$ (all channels are recommended), the random access is used and the probability of selecting channel $m$ is $P_{m}=\frac{1}{M}$.
\end{itemize}
\end{itemize}

To illustrate the channel selection process, let us take the network in Figure \ref{ChannelRec} as an example. Suppose that the branching probability $P_{rec}=0.4$. Since only $R=1$ recommendation is available (i.e., channel 4), the probabilities of choosing the recommended channel 4 and any unrecommended channel are $\frac{0.4}{1}=0.4$ and $\frac{1-0.4}{6-1}=0.12$, respectively.

Numerical studies in \cite{key-2} showed that the static channel recommendation scheme achieves a higher performance over the traditional random channel access scheme without information exchange. However, the fixed value of $P_{rec}$ limits the performance of the static scheme, as explained next.

\subsection{\label{sec:Motivation-Adaptive}Motivations For Adaptive Channel Recommendation}
The static channel recommendation mechanism is simple to implement due to a fixed value of $P_{rec}$. However, it may lead to significant congestions when the number of recommended channels is small. In the extreme case when only $R=1$ channel is recommended, calculation (\ref{eq:SCR-1}) suggests that every user will access that channel with a probability $P_{rec}$. When the number of users $N$ is large, the expected number of users accessing this channel $NP_{rec}$ will be high. Thus heavy congestion happens and each secondary user will get a low expected throughput.

A better way is to adaptively change the value of $P_{rec}$ based on the number of recommended channels. This is the key idea of our proposed algorithm.
To illustrate the advantage of adaptive algorithms, let us first consider a simple heuristic  adaptive algorithm in a homogeneous channel environment, i.e., for each channel $m$, its data rate $B_{m}=B$ and channel state changing probabilities $p_{m}=p,q_{m}=q$. In this algorithm, we choose the branching probability such that the expected number of secondary users choosing a single recommended channel is one. To achieve this, we need to set $P_{rec}$ as in Lemma \ref{lemma12s}.
\begin{lem}\label{lemma12s}
If we choose the branching probability $P_{rec}=\frac{R}{N}$, then the expected number of secondary users choosing
any one of the $R$ recommended channels is one.
\end{lem}

Due to space limitations, we give the detailed proof of Lemma \ref{lemma12s}  in \cite{key-21}. Without going through detailed analysis, it is straightforward to show the benefit for such adaptive approach through simple numerical examples. Let us consider a network with $M=10$ channels and $N=5$ secondary users. For each channel $m$, the initial channel state probability vector is $\boldsymbol{p}_{m}(0)=(0,1)$
and the transition matrix is \[
\Gamma_{m}=\left[\begin{array}{cc}
1-0.01\epsilon & 0.01\epsilon\\
0.01\epsilon & 1-0.01\epsilon\end{array}\right],\]
where $\epsilon$ is called the dynamic factor. A larger value of $\epsilon$ implies that the channels are more dynamic over time.
We are interested in time average system throughput
$U=\frac{\sum_{t=1}^{T}\sum_{n=1}^{N}u_{n}(t)}{T},$
where $u_{n}(t)$ is the throughput of user $n$ at time slot $t$. In the simulation, we set the total number of time slots $T=2000$.

We implement the following three channel access schemes:
\begin{itemize}
\item Random access scheme: each secondary user selects a channel randomly.
\item Static channel recommendation scheme as in \cite{key-2} with the \emph{optimal} constant branching
probability $P_{rec}=0.7$.
\item Heuristic adaptive channel recommendation scheme with the variable branching probability
$P_{rec}=\frac{R}{N}$.
\end{itemize}

Figure \ref{fig:Comparison-of-three} shows that the heuristic adaptive channel recommendation scheme outperforms the static channel recommendation scheme, which in turn outperforms the
random access scheme. Moreover, the heuristic adaptive scheme is more robust to the dynamic channel environment, as it decreases slower than the static scheme when $\epsilon$ increases.

\begin{figure}[ht]
\begin{center}
\includegraphics[scale=0.45]{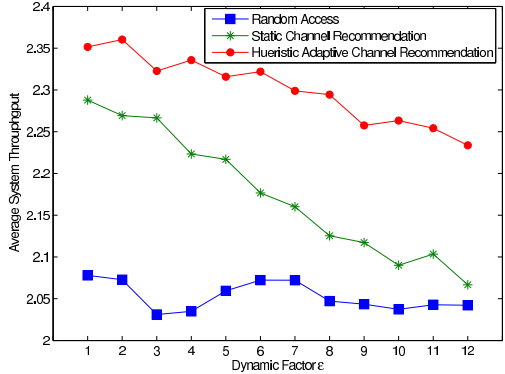}
\end{center}
\caption{\label{fig:Comparison-of-three}Comparison of three channel access
schemes}

\end{figure}

We can imagine that an optimal adaptive scheme (by setting the right $P_{rec}(t)$ over time) can further increase the network performance. However, computing the optimal branching probability in closed-form is very difficult. In the rest of the paper, we will focus on characterizing the structures of the optimal spectrum access strategy and designing an efficient algorithm to achieve the optimum.


\section{\label{sec:Adaptive-Channel-Recommendation}Adaptive Channel Recommendation
Scheme}

We first study the optimal channel recommendation in the homogeneous channel environment, i.e., each channel $m$ has the same data rate $B_{m}=B$ and identical channel state changing probabilities $p_{m}=p,q_{m}=q$. The generalization to the heterogeneous channel setting will be discussed in Section \ref{sec:Adaptive-Channel-RecommendationII}. To find the optimal adaptive spectrum access strategy, we formulate the system as a Markov Decision Process
(MDP). For the sake of simplicity, we assume that the recommendation
buffer size $W=1$, i.e., users only consider the recommendations received in the last time slot. Our method also applies to the case when $W>1$ by using
a high-order MDP formulation, although the analysis is more involved.

\begin{table*}[tt]
\begin{align}
P{}_{R,R'}^{P_{rec}}= & \sum_{m_{r}+m_{u}=R'\ }\sum_{R\geq\bar{m}_{r}\geq m_{r},M-R\geq\bar{m}_{u}\geq m_{u}\ }\sum_{n_{r}+n_{u}=N,n_{r}\geq\bar{m}_{r},n_{u}\geq\bar{m}_{u}}\left(\begin{array}{c}
N\\
n_{r}\end{array}\right)P_{rec}^{n_{r}}(1-P_{rec})^{n_{u}} \nonumber \\
 & \cdot\left(\begin{array}{c}
\bar{m}_{r}\\
m_{r}\end{array}\right)(1-q)^{m_{r}}q^{\bar{m}_{r}-m_{r}}\frac{R!}{(R-\bar{m}_{r})!}\left(\begin{array}{c}
n_{r}-1\\
\bar{m}_{r}-1\end{array}\right)R{}^{-n_{r}} \nonumber \\
 & \cdot\left(\begin{array}{c}
\bar{m}_{u}\\
m_{u}\end{array}\right)(\frac{p}{p+q})^{m_{u}}(\frac{q}{p+q})^{\bar{m}_{u}-m_{u}}\frac{(M-R)!}{(M-R-\bar{m}_{u})!}\left(\begin{array}{c}
n_{u}-1\\
\bar{m}_{u}-1\end{array}\right)(M-R){}^{-n_{u}}. \label{eq:2565} \end{align}
\hrule
\end{table*}

\subsection{MDP Formulation For Adaptive Channel Recommendation}
We model the system as a MDP as follows:
\begin{itemize}
\item \emph{System state}: $R\in\mathcal{R}\triangleq\{0,1,...,\min\{M,N\}\}$ denotes the number of recommended channels at the end of time slot $t$. Since we assume that all channels are statistically identical, then there is no need to keep track of the recommended channel IDs\footnote{Users need to know the IDs of the recommended channels in order to access them. However, the IDs are not important in terms of MDP analysis.}.
\item \emph{Action}: $P_{rec}\in\mathcal{P}\triangleq(0,1)$ denotes the branching probability of choosing the set of recommended channels.
\item \emph{Transition probability}: The probability that action $P_{rec}$ in system state $R$ in time slot $t$ will lead to system state $R'$ in the next time
slot  is
\[P{}_{R,R'}^{P_{rec}}=Pr\{R(t+1)=R'|R(t)=R,P_{rec}(t)=P_{rec}\}.\]
We can compute this probability as in (\ref{eq:2565}), with detailed derivations given in Appendix \ref{Derivation}.
\item \emph{Reward}: $U(R,P_{rec})$ is the expected system throughput in the next time slot when the action $P_{rec}$ is taken under the current system state $R$, i.e.,
\[
U(R,P_{rec})=\sum_{R\in\mathcal{R'}}P{}_{R,R'}^{P_{rec}}U_{R'},\]
where $U_{R'}$ is the system throughput in state $R'$. If $R'$ idle channels are utilized by the secondary users in a time slot, then these $R'$ channels will be recommended at the end of the time slot. Thus, we have \[
U_{R'}=R'B.\]
Recall that $B$ is the data rate that a single user can obtain on an
idle channel.
\item \emph{Stationary Policy:} $\pi\in\Omega\triangleq\mathcal{P}{}^{|\mathcal{R}|}$ maps each state $R$ to an action $P_{rec}$, i.e., $\pi(R)$ is the action $P_{rec}$ taken when the system is in state $R$. \rev{The mapping is stationary and does not depend on time $t$.}
\end{itemize}

Given a stationary policy $\pi$ and the initial state $R_{0}\in\mathcal{R}$,
we define the network's value function as the time average system throughput,
i.e. \[
\Phi_{\pi}(R_{0})=\lim_{T\rightarrow\infty}\frac{1}{T}E_{\pi}\left[\sum_{t=0}^{T-1}U(R(t),\pi(R(t)))\right].\]
We want to find an optimal stationary policy $\pi^{*}$ that maximizes the value function $\Phi_{\pi}(R_{0})$  for any initial state $R_{0}$, i.e.\[
\pi^{*}=\arg\max_{\pi}\Phi_{\pi}(R_{0}),\forall R_{0}\in\mathcal{R}.\]
\rev{Notice that this is a system wide optimization, although the optimal solution can be implemented in a distributed fashion. This is because every user knows the number of recommended channels $R$, and it can determine the same optimal access probability locally.} For example, each user can calculate the optimal spectrum access policy off-line, and determine the real-time optimal channel access probability $P_{rec}$ locally by observing the number of recommended channels $R$ after entering the network.

\subsection{Existence of Optimal Stationary Policy}

MDP formulation above is an average reward based MDP. We can prove that an optimal stationary policy that is independent of initial system state always exists in our MDP formulation. The proof relies on the following lemma from \cite{key-4}.
\begin{lem} \label{lemmaA}
If the state space is finite and every stationary policy leads
to an irreducible Markov chain, then there exists a stationary policy
that is optimal for the average reward based MDP.
\end{lem}
The irreducibility of Markov chain means that it is possible to get to any state from any state. For the adaptive channel recommendation scheme, we have
\begin{lem}\label{lemmaB}
Given a stationary policy $\pi$ for the adaptive channel recommendation
MDP, the resulting Markov chain is irreducible.
\end{lem}
\begin{proof}
We consider the following two cases:

Case I, when $0<q<1$: since $0<P_{rec}<1$, $0<p\leq1$, and $0<q<1$, we can verify that given any state $R$, the transition probability $P{}_{R,R'}^{P_{rec}}>0$ for all $R'\in\mathcal{R}$. Thus, any two states communicate with each other.

Case II, when $q=1$: for all $R\in\mathcal{R}$, the transition probability $P{}_{R,R'}^{P_{rec}}>0$ if $R'\in\{0,...,\min\{M-R,N\}\}$. It follows that the state $R'=0$ is accessible from any other state
$R\in\mathcal{R}$. By setting $R=0$, we
see that $P{}_{R,R'}^{P_{rec}}>0$,  for all $R'\in\{0,...,\min\{M,N\}\}$. That
is, any other state $R'\in\mathcal{R}$ is also accessible from the state
$R=0$. Thus, any two states communicate with each other.

Since any two states communicate with each other in all cases and
the number of system state $|\mathcal{R}|$ is finite, the resulting
Markov chain is irreducible.
\end{proof}
Combining Lemmas \ref{lemmaA} and \ref{lemmaB}, we have
\begin{thm}
There exists an optimal stationary policy for the adaptive channel
recommendation MDP.
\end{thm}

Furthermore, the irreducibility of the adaptive channel
recommendation MDP also implies that the optimal stationary policy $\pi^{*}$ is independent of the initial state $R_{0}$ \cite{key-4}, i.e.
\[
\Phi_{\pi^{*}}(R_{0})=\Phi_{\pi^{*}}, \forall R_{0}\in\mathcal{R},
\]where $\Phi_{\pi^{*}}$ is the maximum time average system throughput.
In the rest of the paper, we will just use \textquotedblleft{}optimal policy\textquotedblright{} to refer \textquotedblleft{}optimal stationary policy that is independent of the initial system state\textquotedblright{}.

\subsection{Structure of Optimal Stationary Policy}

Next we characterize the structure of the optimal policy without using the closed-form expressions of the policy (which is generally hard to achieve). The key idea is to treat the average
reward based MDPs as the limit of a sequence of discounted reward MDPs with discounted
factors going to one. Under the irreducibility condition, the average reward based MDP thus inherits the
structure property from the corresponding discounted reward MDP \cite{key-4}.
We can write down the Bellman equations of the discounted version of our MDP problem as:
\begin{equation}
V_{t}(R)=\max_{P_{rec}\in\mathcal{P}}\sum_{R'\in\mathcal{R}}P{}_{R,R'}^{P_{rec}}[U_{R'}+\beta V_{t+1}(R')],\forall R\in\mathcal{R},\label{eq:5647}\end{equation}
where $V_{t}(R)$ is the discounted maximum expected system throughput starting from time
slot $t$ when the system in state $R$.

Due to the combinatorial complexity of the transition probability
$P{}_{R,R'}^{P_{rec}}$ in (\ref{eq:2565}), it is difficult to obtain
the structure results for the general case. We further limit our attention to the following two asymptotic cases.

\subsubsection{Case One, the number of channels $M$ goes to infinity while the number of users $N$ stays finite}

In this case, the number of channels is much larger than the number of secondary users, and thus heavy congestion rarely happens on any channel. Thus it is safe to emphasizing on accessing the recommended channels. Before proving the main result of Case One in Theorem \ref{thmM>N}, let us first characterize the property of discounted maximum expected system payoff $V_t(R)$.
\begin{proposition}
\label{lem:When-,-i.e.II}When $M=\infty$ and $N<\infty$
, the value function $V_{t}(R)$ for the discounted adaptive
channel recommendation MDP is nondecreasing in $R$ .\end{proposition}

The proof of Proposition \ref{lem:When-,-i.e.II} is given in the Appendix.
Based on the monotone property of the value function $V_{t}(R)$, we prove the following main result.
\begin{thm}
\label{thmM>N}
When $M=\infty$ and $N<\infty$, for the adaptive channel
recommendation MDP, the optimal stationary policy $\pi^{*}$ is monotone,
that is, $\pi^{*}(R)$ is nondecreasing on $R\in\mathcal{R}$.\end{thm}
\begin{proof}
For the ease of discussion, we define \[
Q_{t}(R,P_{rec})=\sum_{R'\in\mathcal{R}}P{}_{R,R'}^{P_{rec}}[U_{R'}+\beta V_{t+1}(R')],\]
with the partial cross derivative being \begin{eqnarray*}
\frac{\partial^{2}Q_{t}(R,P_{rec})}{\partial R\partial P_{rec}} & = & \frac{\partial\sum_{R'\in\mathcal{R}}P{}_{R+1,R'}^{P_{rec}}[U_{R'}+\beta V_{t+1}(R')]}{\partial P_{rec}}\\
 &  & -\frac{\partial\sum_{R'\in\mathcal{R}}P{}_{R,R'}^{P_{rec}}[U_{R'}+\beta V_{t+1}(R')]}{\partial P_{rec}}.\end{eqnarray*}
By Lemma $6$ in the Appendix, we know the reverse cumulative distribution
function $\sum_{R'\in\mathcal{R}}P{}_{R,R'}^{P_{rec}}$ is supermodular on $\mathcal{R}\times\mathcal{P}$.
It implies \[
\frac{\partial\sum_{R'\in\mathcal{R}}P{}_{R+1,R'}^{P_{rec}}}{\partial P_{rec}}-\frac{\partial\sum_{R'\in\mathcal{R}}P{}_{R,R'}^{P_{rec}}}{\partial P_{rec}}\geq0.\]
Since $V_{t+1}(R')$ is nondecreasing in $R'$ by Proposition \ref{lem:When-,-i.e.II} and  $U_{R'}=R'B$,   we know that $U_{R'}+\beta V_{t+1}(R')$ is also nondecreasing in $R'$.
Then we have\begin{eqnarray*}
 &  & \frac{\partial\sum_{R'\in\mathcal{R}}P{}_{R+1,R'}^{P_{rec}}[U_{R'}+\beta V_{t+1}(R')]}{\partial P_{rec}}\\
 & \geq & \frac{\partial\sum_{R'\in\mathcal{R}}P{}_{R,R'}^{P_{rec}}[U_{R'}+\beta V_{t+1}(R')]}{\partial P_{rec}},\end{eqnarray*}
i.e.,\[
\frac{\partial^{2}Q_{t}(R,P_{rec})}{\partial R\partial P_{rec}}\geq0,\]
which implies that $Q_{t}(R,P_{rec})$ is supermodular on $\mathcal{R}\times\mathcal{P}$.
Since \[
\pi^{*}(R)=\arg\max_{P_{rec}}Q_{t}(R,P_{rec}),\]
by the property of super-modularity, the optimal policy $\pi^{*}(R)$
is nondecreasing on $R$ for the discounted MDP above. Since the average
reward based MDP inherits its structure property, this result is also
true for the adaptive channel recommendation MDP.
\end{proof}

\subsubsection{Case Two, the number of users $N$ goes to infinity while the number of channels $M$ stays finite}
In this case, the number of secondary users is much larger than the number of channels, and thus congestion becomes a major concern. However, since there are infinitely many secondary users, all the idle channels at each time slot can be utilized as long as users have positive probabilities to access all channels. From the system's point of view, the cognitive radio network operates in the saturation state. Formally, we show that
\begin{thm}\label{InfN}
When $N=\infty$ and $M<\infty$, for the adaptive channel
channel recommendation MDP, any stationary policy $\pi$ satisfying\[
0<\pi(R)<1,\forall R\in\mathcal{R},\]
is optimal. \end{thm}
\begin{proof}
We first define the sets of policies $\Delta\triangleq\{\pi:0<\pi(R)<1,\forall R\in\mathcal{R}\}$
and $\Delta^{c}=\Omega\backslash\Delta$.  Recall that the value of $\pi(R)$ equals the probability of choosing the set of recommended channels, i.e., $P_{rec}$.

Then it is easy to check that the probability of accessing an arbitrary channel
$m$ is positive under any policy $\pi\in\Delta$. Since
the number of secondary users $N=\infty$, it implies that all the
channels will be accessed by the secondary users. In this case, the
transition probability from a system state $R$ to $R'$ of the resulting
Markov chain is given by\begin{eqnarray}
 &  & P_{R,R'}^{\pi(R)}\nonumber \\
 & = & \sum_{m_{r}+m_{u}=R',m_{r}\le R,m_{u}\le M-R}\left(\begin{array}{c}
R\\
m_{r}\end{array}\right)(1-q)^{m_{r}}q^{R-m_{r}}\nonumber\\
 &  & \cdot\left(\begin{array}{c}
M-R\\
m_{u}\end{array}\right)(\frac{p}{p+q})^{m_{u}}(\frac{q}{p+q})^{M-R-m_{u}},\label{eq:215.1}\end{eqnarray}
which is independent of the branching probability $\pi(R)$. It implies
that any policy $\pi\in\Delta$ leads to a Markov chain with
the same transition probabilities $P_{R,R'}^{P_{rec}}$. Thus, any
policy $\pi\in\Delta$ offers the same time average system throughput.

We next show that any policy $\pi'\in\Delta^{c}$  leads to a payoff no better than the payoff of a policy $\pi\in\Delta$. For a policy $\pi'$ where there exists some states $\bar{R}$
such that $\pi'(\bar{R})=0$, the transition probability from
the system state $\bar{R}$ to $R'$ is\begin{eqnarray*}
P_{\bar{R},R'}^{\pi'(\bar{R})} & = & \begin{cases}
\left(\begin{array}{c}
M-\bar{R}\\
R'\end{array}\right)(\frac{p}{p+q})^{R'}(\frac{q}{p+q})^{M-\bar{R}-R'}\\
\ \ \ \ \ \ \ \ \ \ \ \ \ \ \ \ \ \ \ \ \ \ \ \ \ \ \mbox{If}\ R'\le M-\bar{R},\\
0\mbox{\ \ \ \ \ \ \ \ \ \ \ \ \ \ \ \ \ \ \ \ \ \ \ \ \ \ If}\ R'>M-\bar{R}.
\end{cases}\end{eqnarray*}
If there exists some states $\hat{R}$ such that $\pi'(\hat{R})=1$,
we have the transition probability as\begin{eqnarray*}
P_{\hat{R},R'}^{\pi'(\hat{R})} & = & \begin{cases}
\left(\begin{array}{c}
\hat{R}\\
R'\end{array}\right)(1-q)^{R'}q^{\hat{R}-R'} & \mbox{If\ }R'\le\hat{R},\\
0 & \mbox{If\ }R'>\hat{R}.\end{cases}\end{eqnarray*}
Since \begin{eqnarray*}
 &  & \left(\begin{array}{c}
M-\bar{R}\\
R'\end{array}\right)(\frac{p}{p+q})^{R'}(\frac{q}{p+q})^{M-\bar{R}-R'}\\
 & = & \sum_{j=0}^{\bar{R}}\left(\begin{array}{c}
\bar{R}\\
j\end{array}\right)(1-q)^{j}q^{\bar{R}-j}\\
 &  & \cdot\left(\begin{array}{c}
M-\bar{R}\\
R'\end{array}\right)(\frac{p}{p+q})^{R'}(\frac{q}{p+q})^{M-\bar{R}-R'},\end{eqnarray*}
and\begin{eqnarray*}
 &  & \left(\begin{array}{c}
\hat{R}\\
R'\end{array}\right)(1-q)^{R'}q^{\hat{R}-R'}\\
 & = & \sum_{j=0}^{M-\hat{R}}\left(\begin{array}{c}
M-\hat{R}\\
j\end{array}\right)(\frac{p}{p+q})^{j}(\frac{q}{p+q})^{M-\hat{R}-j}\\
 &  & \cdot\left(\begin{array}{c}
\hat{R}\\
R'\end{array}\right)(1-q)^{R'}q^{\hat{R}-R'},\end{eqnarray*}
compared with (\ref{eq:215.1}), we have\[
\sum_{R'=i}^{M}P_{R,R'}^{\pi(R)}\geq\sum_{R'=i}^{M}R_{R,R'}^{\pi'(R)},\forall i,R\in\mathcal{R},\pi\in\Delta,\pi'\in\Delta^{c}.\]

Suppose that the time horizon consists of any $T$ time slots, and $V_{t}^{\pi}(R)$
denotes the expected system throughput under the policy $\pi$ by starting from
time slot $t$ when the system in state $R$.

When $t=T$, \begin{eqnarray*}
V_{T}^{\pi}(R) & = & V_{T}^{\pi'}(R)\\
 & = & U_{R}\\
 & = & RB,\forall R\in\mathcal{R},\pi\in\Delta,\pi'\in\Delta^{c}.\end{eqnarray*}
It follows that $U_{R}+\beta V_{T}^{\pi}(R)=U_{R}+\beta V_{T}^{\pi'}(R)$,
and hence \begin{eqnarray*}
 &  & \sum_{R'=0}^{M}P_{R,R'}^{\pi(R)}[U(R)+\beta V_{T}^{\pi}(R)]\\
 & \geq & \sum_{R'=0}^{M}R_{R,R'}^{\pi'(R)}[U(R)+\beta V_{T}^{\pi'}(R)],\end{eqnarray*}
i.e., \[
V_{T-1}^{\pi}(R)\geq V_{T-1}^{\pi'}(R),\forall R\in\mathcal{R},\pi\in\Delta,\pi'\in\Delta^{c}.\]
Recursively, for any time slots $t\leq T$, we can show that\[
V_{t}^{\pi}(R)\geq V_{t}^{\pi'}(R),\forall R\in\mathcal{R},\pi\in\Delta,\pi'\in\Delta^{c}.\]
Thus, if there exists a policy $\pi'\in\Delta^{c}$ that is optimal, then
all the policies $\pi\in\Delta$ is also optimal. If there does not
exist such a policy $\pi'$, then we conclude that only the policy
$\pi\in\Delta$ is optimal.
\end{proof}


\section{\label{sec:Model-Reference-Adaptive}Model Reference Adaptive Search
For Optimal Spectrum Access Policy}
Next we will design an algorithm that can converge to the optimal policy under general system parameters (not limiting to the two asymptotic cases).
Since the action space of the adaptive channel recommendation MDP
is continuous (i.e., choosing a probability $P_{rec}$ in $(0,1)$), the traditional method of discretizing the action space
followed by the policy, value iteration, or Q-learning cannot guarantee to converge to the optimal policy. To overcome this difficulty, we propose a new algorithm developed from the Model Reference Adaptive Search method, which was recently developed in the Operations Research community \cite{key-5}. We will show that the proposed algorithm is easy to implement and is provably convergent to the optimal policy.

\subsection{Model Reference Adaptive Search Method}
We first introduce the basic idea of the Model Reference Adaptive Search (MRAS) method. Later on, we will show how the method can be used to obtain optimal spectrum access policy for our problem.

The MRAS method is a new randomized
method for global optimization \cite{key-5}. The key idea is to randomize the original optimization problem over the feasible region according
to a specified probabilistic model. The method then generates candidate solutions
and updates the probabilistic model on the basis of elite solutions
and a reference model, so that to guide the future search toward better solutions.

Formally, let $J(x)$ be the objective function to maximize. The
MRAS method is an iterative algorithm, and it includes three phases in each iteration $k$:
\begin{itemize}
\item \emph{Random solution generation}: generate a set of random solutions \rev{$\{x\}$ in the feasible set $\chi$} according to a parameterized probabilistic model $f(x,v_{k})$, which is a probability density function (pdf) with parameter $v_{k}$. \rev{The number of solutions to generate is a fixed system parameter.}
\item \emph{Reference distribution construction}: select elite solutions among the randomly generated set in the previous phase, such that the chosen ones satisfy $J(x)\geq\gamma$. Construct a reference probability distribution as\begin{eqnarray}
g_{k}(x) & =\begin{cases}
\frac{I_{\{J(x)\geq\gamma\}}}{E_{f(x,v_{0})}[\frac{I_{\{J(x)\geq\gamma\}}}{f(x,v_{0})}]} & k=1,\\
\frac{e^{J(x)}I_{\{J(x)\geq\gamma\}}g_{k-1}(x)}{E_{g_{k-1}}[e^{J(x)}I_{\{J(x)\geq\gamma\}}]} & k\geq2,\end{cases}\label{eq:2546-1}\end{eqnarray}
where $I_{\{\varpi\}}$ is an indicator function, which equals $1$ if the event $\varpi$ is true and zero otherwise. Parameter $v_{0}$ is the initial parameter for the probabilistic model (used during the first iteration, i.e., $k=1$),
and $g_{k-1}(x)$ is the reference distribution in the previous iteration (used when $k\geq 2$).
\item \emph{Probabilistic model update}: update the parameter $v$ of the probabilistic
model $f(x,v)$ by minimizing the Kullback-Leibler divergence between
$g_{k}(x)$ and $f(x,v)$, i.e.\begin{equation}
v_{k+1}=\arg\min_{v}E_{g_{k}}\left[\ln\frac{g_{k}(x)}{f(x,v)}\right].\label{eq:12567}\end{equation}

\end{itemize}

By constructing the reference distribution according to
(\ref{eq:2546-1}), the expected performance of random elite solutions can be improved under the new reference distribution, i.e.,\begin{eqnarray}
E_{g_{k}}[e^{J(x)}I_{\{J(x)\geq\gamma\}}] & = & \frac{\int_{x\in\chi}e^{2J(x)}I_{\{J(x)\geq\gamma\}}g_{k-1}(x)dx}{E_{g_{k-1}}[e^{J(x)}I_{\{J(x)\geq\gamma\}}]}\nonumber \\
 & = & \frac{E_{g_{k-1}}[e^{2J(x)}I_{\{J(x)\geq\gamma\}}]}{E_{g_{k-1}}[e^{J(x)}I_{\{J(x)\geq\gamma\}}]}\nonumber \\
 & \geq & E_{g_{k-1}}[e^{J(x)}I_{\{J(x)\geq\gamma\}}].\label{eq:2354}\end{eqnarray}
To find a better solution to the optimization problem, it is natural to update the probabilistic model (from which random solution  are generated in the first stage) as close to the new reference probability as possible, as done in the third stage.

\subsection{Model Reference Adaptive Search For Optimal Spectrum Access Policy}

In this section, we design an algorithm based on the MRAS method to find the optimal
spectrum access policy. Here we treat the adaptive channel recommendation MDP as a global optimization problem over the policy space. The key challenge is the choice of proper probabilistic model $f(\cdot)$, which is crucial for the convergence of the MRAS algorithm.

\subsubsection{Random Policy Generation}

To apply the MRAS method, we first need to set up a random policy generation
mechanism. Since the action space of the channel recommendation MDP
is continuous, we use the Gaussian distributions.
Specifically, we generate sample actions $\pi(R)$ from a Gaussian
distribution for each system state $R\in\mathcal{R}$ independently,
i.e. $\pi(R)\sim\mathcal{N}(\mu_{R},\sigma_{R}^{2})$.\footnote{Note that the Gaussian distribution has a support over $(-\infty,+\infty)$, which is larger than the feasible region of $\pi(R)$. This issue will be handled in Section \ref{STE}.} In this case,
a candidate policy $\pi$ can be generated from the joint distribution
of $|\mathcal{R}|$ independent Gaussian distributions, i.e.,\begin{eqnarray*}
(\pi(0),...,\pi(\min\{M,N\})) & \sim & \mathcal{N}(\mu_{0},\sigma_{0}^{2})\times\cdots\\
 &  & \times\mathcal{N}(\mu_{\min\{M,N\}},\sigma_{\min\{M,N\}}^{2}).\end{eqnarray*}
As shown later, Gaussian distribution has nice analytical and convergent
properties for the MRAS method.

For the sake of brevity, we denote $f(\pi(R),\mu_{R},\sigma_{R})$
as the pdf of the Gaussian distribution $\mathcal{N}(\mu_{R},\sigma_{R}^{2})$,
and $f(\pi,\boldsymbol{\mu},\boldsymbol{\sigma})$ as random policy
generation mechanism with parameters $\boldsymbol{\mu}\triangleq(\mu_{0},...,\mu_{\min\{M,N\}})$ and $\boldsymbol{\sigma}\triangleq(\sigma_{0},...,\sigma_{\min\{M,N\}})$,
i.e.,\begin{eqnarray*}
f(\pi,\boldsymbol{\mu},\boldsymbol{\sigma}) & = & \prod_{R=0}^{\min\{M,N\}}f(\pi(R),\mu_{R},\sigma_{R})\\
 & = & \prod_{R=0}^{\min\{M,N\}}\frac{1}{\sqrt{2\varphi\sigma_{R}^{2}}}e^{-\frac{(\pi(R)-\mu_{R})^{2}}{2\sigma_{R}^{2}}},\end{eqnarray*}
where $\varphi$ is the circumference-to-diameter ratio. \com{change $pi$ to a single Greek letter.}

\subsubsection{System Throughput Evaluation}\label{STE}

Given a candidate policy $\pi$ randomly generated based on $f(\pi,\boldsymbol{\mu},\boldsymbol{\sigma})$, we need to evaluate the expected system
throughput $\Phi_{\pi}$. From (\ref{eq:2565}), we obtain the
transition probabilities $P{}_{R,R'}^{\pi(R)}$ for any system state
$R,R'\in\mathcal{R}$. Since a policy $\pi$ leads to a finitely
irreducible Markov chain, we can obtain its stationary distribution.
Let us denote the transition matrix of the Markov chain as $Q\triangleq[P{}_{R,R'}^{\pi(R)}]_{|\mathcal{R}|\times|\mathcal{R}|}$
and the stationary distribution as $\boldsymbol{p}=(Pr(0),...,Pr(\min\{M,N\}))$.
Obviously, the stationary distribution can be obtained by solving
the following equation\[
\boldsymbol{p}Q=\boldsymbol{p}.\]
We then calculate the expected system throughput $\Phi_{\pi}$ by
\[
\Phi_{\pi}=\sum_{R\in\mathcal{R}}Pr(R)U_{R}.\]

Note that in the discussion above, we assume that $\pi\in\Omega$
implicitly, where $\Omega$ is the feasible policy space. Since Gaussian distribution has a support over
$(-\infty,+\infty)$, we thus extend the definition of expected system throughput $\Phi_{\pi}$ over $(-\infty,+\infty)^{|\mathcal{R}|}$
as\[
\Phi_{\pi}=\begin{cases}
\sum_{R\in\mathcal{R}}Pr(R)U_{R} & \pi\in\Omega,\\
-\infty & \mbox{Otherwise.}\end{cases}\]
In this case, whenever any generated policy $\pi$ is not feasible, we have $\Phi_{\pi}=-\infty$.  \rev{As a result, such policy $\pi$ will not be selected as an elite sample (discussed next) and will not used for probability updating.} Hence the search of MRAS algorithm will not bias towards any unfeasible policy space.

\subsubsection{Reference Distribution Construction}

To construct the reference distribution, we first need to select the
elite policies. Suppose $L$ candidate policies, $\pi_{1},\pi_{2},...,\pi_{L}$,
are generated at each iteration. We order them based on an increasing order of the expected system throughputs
$\Phi_{\pi}$, i.e., $\Phi_{\hat{\pi}_{1}}\le\Phi_{\hat{\pi}_{2}}\le...\le\Phi_{\hat{\pi}_{L}}$, and set the elite threshold as
\begin{equation*}
\gamma=\Phi_{\hat{\pi}_{\lceil(1-\rho)L\rceil}},
\end{equation*}
where $0<\rho<1$ is the elite ratio. \rev{For example, when $L=100$ and $\rho=0.4$, then $\gamma=\Phi_{\hat{\pi}_{60}}$ and the last $40$ samples in the sequence will be selected as elite samples.}
Note that as long as $L$ is sufficiently large, we shall have $\gamma<\infty$ and hence only feasible policies $\pi$ are selected. According to (\ref{eq:2546-1}),
we then construct the reference distribution as

\begin{eqnarray}
g_{k}(\pi) & =\begin{cases}
\frac{I_{\{\Phi_{\pi}\geq\gamma\}}}{E_{f(\pi,\boldsymbol{\mu}_{0},\boldsymbol{\sigma}_{0})}[\frac{_{I_{\{\Phi_{\pi}\geq\gamma\}}}}{f(\pi,\boldsymbol{\mu}_{0},\boldsymbol{\sigma}_{0})}]} & k=1,\\
\frac{e^{\Phi_{\pi}}I_{\{\Phi_{\pi}\geq\gamma\}}g_{k-1}(\pi)}{E_{g_{k-1}}[e^{\Phi_{\pi}}I_{\{\Phi_{\pi}\geq\gamma\}}]} & k\geq2.\end{cases}\label{eq:8975}\end{eqnarray}

\subsubsection{Policy Generation Update}

For the MRAS algorithm, the critical issue is the updating of random
policy generation mechanism $f(\pi,\boldsymbol{\mu},\boldsymbol{\sigma})$,
or solving the problem in (\ref{eq:12567}). The optimal update rule is described as follow.
\begin{thm}
The optimal parameter $(\boldsymbol{\mu},\boldsymbol{\sigma})$ that
minimizes the Kullback-Leibler divergence between the reference distribution
$g_{k}(\pi)$ in (\ref{eq:8975}) and the new policy generation mechanism
$f(\pi,\boldsymbol{\mu},\boldsymbol{\sigma})$ is \begin{eqnarray}
\mu_{R} & = & \frac{\int_{\pi\in\Omega}e^{(k-1)\Phi_{\pi}}I_{\{\Phi_{\pi}\geq\gamma\}}\pi(R)d\pi}{\int_{\pi\in\Omega}e^{(k-1)\Phi_{\pi}}I_{\{\Phi_{\pi}\geq\gamma\}}d\pi},\forall R\in\mathcal{R},\label{eq:524-1}\\
\sigma_{R}^{2} & = & \frac{\int_{\pi\in\Omega}e^{(k-1)\Phi_{\pi}}I_{\{\Phi_{\pi}\geq\gamma\}}[\pi(R)-\mu_{R}]^{2}d\pi}{\int_{\pi\in\Omega}e^{(k-1)\Phi_{\pi}}I_{\{\Phi_{\pi}\geq\gamma\}}d\pi},\forall R\in\mathcal{R}.\nonumber \\ \nonumber
\label{eq:524-2}\\\end{eqnarray}
\end{thm}
\linespread{0.9}
\begin{proof}
First, from (\ref{eq:8975}), we have\begin{eqnarray*}
g_{1}(\pi) & = & \frac{I_{\{\Phi_{\pi}\geq\gamma\}}}{E_{f(\pi,\boldsymbol{\mu}_{0},\boldsymbol{\sigma}_{0})}[\frac{_{I_{\{\Phi_{\pi}\geq\gamma\}}}}{f(\pi,\boldsymbol{\mu}_{0},\boldsymbol{\sigma}_{0})}]}\\
 & = & \frac{I_{\{\Phi_{\pi}\geq\gamma\}}}{\int_{\pi\in\Omega}I_{\{\Phi_{\pi}\geq\gamma\}}d\pi},\end{eqnarray*}
and,
\begin{eqnarray*}
g_{2}(\pi) & = & \frac{e^{\Phi_{\pi}}I_{\{\Phi_{\pi}\geq\gamma\}}g_{1}(\pi)}{E_{g_{1}}[e^{\Phi_{\pi}}I_{\{\Phi_{\pi}\geq\gamma\}}]}\\
 & = & \frac{e^{\Phi_{\pi}}I_{\{\Phi_{\pi}\geq\gamma\}}I_{\{\Phi_{\pi}\geq\gamma\}}}{E_{g_{1}}[e^{\Phi_{\pi}}I_{\{\Phi_{\pi}\geq\gamma\}}]\int_{\pi\in\Omega}I_{\{\Phi_{\pi}\geq\gamma\}}d\pi}\\
 & = & \frac{e^{\Phi_{\pi}}I_{\{\Phi_{\pi}\geq\gamma\}}I_{\{\Phi_{\pi}\geq\gamma\}}}{\int_{\pi\in\Omega}e^{\Phi_{\pi}}I_{\{\Phi_{\pi}\geq\gamma\}}\frac{I_{\{\Phi_{\pi}\geq\gamma\}}}{\int_{\pi\in\Omega}I_{\{\Phi_{\pi}\geq\gamma\}}d\pi}d\pi\int_{\pi\in\Omega}I_{\{\Phi_{\pi}\geq\gamma\}}d\pi}\\
 & = & \frac{e^{\Phi_{\pi}}I_{\{\Phi_{\pi}\geq\gamma\}}}{\int_{\pi\in\Omega}e^{\Phi_{\pi}}I_{\{\Phi_{\pi}\geq\gamma\}}d\pi}.\end{eqnarray*}
Repeat the above computation iteratively, we have%
\begin{equation}
g_{k}(\pi)=\frac{e^{(k-1)\Phi_{\pi}}I_{\{\Phi_{\pi}\geq\gamma\}}}{\int_{\pi\in\Omega}e^{(k-1)\Phi_{\pi}}I_{\{\Phi_{\pi}\geq\gamma\}}d\pi},k\geq1.\label{eq:23589}\end{equation}
Then, the problem in (\ref{eq:12567}) is equivalent to solving \begin{eqnarray}
\max_{\boldsymbol{\mu},\boldsymbol{\sigma}} & \int_{\pi\in\Omega}g_{k}(\pi)\ln f(\pi,\boldsymbol{\mu},\boldsymbol{\sigma})d\pi,\label{eq:235}\\
\mbox{subject to} & \boldsymbol{\mu},\boldsymbol{\sigma}\succeq0,\nonumber\end{eqnarray}
Substituting (\ref{eq:23589}) into (\ref{eq:235}), we have\begin{eqnarray}
\max_{\boldsymbol{\mu},\boldsymbol{\sigma}} & \int_{\pi\in\Omega}e^{(k-1)\Phi_{\pi}}I_{\{\Phi_{\pi}\geq\gamma\}}\ln f(\pi,\boldsymbol{\mu},\boldsymbol{\sigma})d\pi,\label{eq:235-1}\\
\mbox{subject to} & \boldsymbol{\mu},\boldsymbol{\sigma}\succeq0,\nonumber\end{eqnarray}

Function $f(\pi(R),\mu_{R},\sigma_{R})$ is log-concave, since it is the pdf of the Gaussian distribution. Since the log-concavity is closed under multiplication,
then $f(\pi,\boldsymbol{\mu},\boldsymbol{\sigma})=\prod_{R=0}^{\min\{M,N\}}f(\pi(R),\mu_{R},\sigma_{R})$
is also log-concave. It implies the problem in (\ref{eq:235}) is
a concave optimization problem. Solving by the first order condition,
we have \begin{eqnarray*}
\frac{\partial\int_{\pi\in\Omega}e^{(k-1)\Phi_{\pi}}I_{\{\Phi_{\pi}\geq\gamma\}}\ln f(\pi,\boldsymbol{\mu},\boldsymbol{\sigma})d\pi}{\partial\mu_{R}} & = & 0,\forall R\in\mathcal{R},\\
\frac{\partial\int_{\pi\in\Omega}e^{(k-1)\Phi_{\pi}}I_{\{\Phi_{\pi}\geq\gamma\}}\ln f(\pi,\boldsymbol{\mu},\boldsymbol{\sigma})d\pi}{\partial\sigma_{R}} & = & 0,\forall R\in\mathcal{R},\end{eqnarray*}
which leads to (\ref{eq:524-1}) and (\ref{eq:524-2}).
Due to the concavity of the optimization problem in (\ref{eq:235}),
the solution is also the global optimum for the random policy generation updating.
\end{proof}
\linespread{1.0}

\subsubsection{MARS Algorithm For Optimal Spectrum Access Policy}

Based on the MARS algorithm, we generate $L$ candidate polices
at each iteration. Then the updates in (\ref{eq:524-1}) and
(\ref{eq:524-2}) are replaced by the sample average version in (\ref{eq:525-1})
and (\ref{eq:525-2}), respectively. As a summary, we describe the MARS-based
algorithm for finding the optimal spectrum access policy of adaptive
channel recommendation MDP in Algorithm \ref{alg:MRAS-Method-For}.

\subsection{Convergence of Model Reference Adaptive Search}

In this part, we discuss the convergence property of the MRAS-based
optimal spectrum access policy. For ease of exposition, we
assume that the adaptive channel recommendation MDP has a unique global
optimal policy. Numerical studies in \cite{key-5} show that the MRAS
method also converges where there are multiple global optimal solutions.
We shall show that the
random policy generation mechanism $f(\pi,\boldsymbol{\mu}_{k},\boldsymbol{\sigma}_{k})$ will eventually generate the optimal policy.
\begin{thm}
\label{theorem1}
For the MRAS algorithm, \rev{the limiting point of the policy sequence}  $\{\pi_{k}\}$ generated by the sequence
of random policy generation mechanism $\{f(\pi,\boldsymbol{\mu}_{k},\boldsymbol{\sigma}_{k})\}$ converges point-wisely to
the optimal spectrum access policy $\pi^{*}$ for the adaptive channel
recommendation MDP, i.e., \begin{eqnarray}
\lim_{k\rightarrow\infty}E_{f(\pi,\boldsymbol{\mu}_{k},\boldsymbol{\sigma}_{k})}[\pi(R)] & = & \pi^{*}(R),\forall R\in\mathcal{R},\label{eq:lemma12}\\
\lim_{k\rightarrow\infty}Var_{f(\pi,\boldsymbol{\mu}_{k},\boldsymbol{\sigma}_{k})}[\pi(R)] & = & 0,\forall R\in\mathcal{R}.\label{eq:lemma12-2}\end{eqnarray}
\end{thm}
The proof is given in the Appendix.

From Theorem \ref{theorem1}, we see that parameter $(\mu_{R,k},\sigma_{R,k})$
for updating in (\ref{eq:525-1}) and (\ref{eq:525-2}) also converges, i.e.,\begin{eqnarray*}
\lim_{k\rightarrow\infty}\mu_{R,k} & = & \pi^{*}(R),\forall R\in\mathcal{R},\\
\lim_{k\rightarrow\infty}\sigma_{R,k} & = & 0,\forall R\in\mathcal{R}.\end{eqnarray*}
Thus, we can use $\max_{R\in\mathcal{R}}\sigma_{R,k}<\xi$ as the stopping
criterion in Algorithm \ref{alg:MRAS-Method-For}.

\begin{algorithm}[tt]
\begin{algorithmic}[1]
\State \textbf{Initialize} parameters for Gaussian distributions $(\boldsymbol{\mu}_{0},\boldsymbol{\sigma}_{0})$,
the elite ratio $\rho$, and the stopping criterion $\xi$. Set initial elite threshold $\gamma_{0}=0$ and iteration index $k=0$.
\Repeat{:}
\State \textbf{Increase} iteration index $k$ by 1.
\State \textbf{Generate} $L$ candidate policies $\pi_{1},...,\pi_{L}$ from the random policy generation mechanism $f(\pi,\boldsymbol{\mu}_{k-1},\boldsymbol{\sigma}_{k-1})$.
\State \textbf{Select} elite policies by setting the elite threshold $\gamma_{k}=\max\{\Phi_{\hat{\pi}_{\lceil(1-\rho)L\rceil}},\gamma_{k-1}\}.$
\State \textbf{Update} the random policy generation mechanism by \begin{align}
\mu_{R,k}&=\frac{\sum_{i=1}^{L}e^{(k-1)\Phi_{\pi}}I_{\{\Phi_{\pi_{i}}\geq\gamma_{k}\}}\pi_{i}(R)}{\sum_{i=1}^{L}e^{(k-1)\Phi_{\pi}}I_{\{\Phi_{\pi_{i}}\geq\gamma_{k}\}}},&\forall R\in\mathcal{R},\label{eq:525-1}\\
\sigma_{R,k}^{2}&=\frac{\sum_{i=1}^{L}e^{(k-1)\Phi_{\pi}}I_{\{\Phi_{\pi_{i}}\geq\gamma_{k}\}}[\pi_{i}(R)-\mu_{R}]^{2}}{\sum_{i=1}^{L}e^{(k-1)\Phi_{\pi}}I_{\{\Phi_{\pi_{i}}\geq\gamma_{k}\}}},&\forall R\in\mathcal{R}. \nonumber \\ \label{eq:525-2}\end{align}
\Until{$\max_{R\in\mathcal{R}}\sigma_{R,k}<\xi.$}
\end{algorithmic}
\caption{\label{alg:MRAS-Method-For}MRAS-based Algorithm For Adaptive Recommendation Based Optimal Spectrum Access}
\end{algorithm}


\section{\label{sec:Adaptive-Channel-RecommendationII}Adaptive Channel Recommendation With Channel Heterogeneity}
\rev{
We now generalize the adaptive channel recommendation to the heterogeneous
channel setting. Recall that the system state $R$ in the homogeneous
channel case only keeps track of how many
channels are recommended. In a heterogeneous channel environment,
each channel has different a data rate $B_{m}$ and channel state changing
probabilities $p_{m}$ and $q_{m}$. Keeping track of the number of recommend channels is not enough for optimal decision. Intuitively, if a channel
with higher data rate $B_{m}$ is recommended, users should choose this channel with a higher weight. The new system state for the heterogeneous channel case should be defined as a vector $\vec{R}\triangleq(I_{1},...,I_{M})$,
where $I_{m}=1$ if channel $m$ is recommended and $I_{m}=0$ otherwise.
The objective of the heterogeneous channel recommendation MDP is then
to find the optimal channel access probabilities $\{P_{m}(\vec{R})\}_{m=1}^{M}$
for each system state $\vec{R}$ where $P_{m}(\vec{R})$ is the probability
of selecting channel $m$. 

Similarly with the homogeneous channel
case, we can apply the MRAS method to obtain the optimal
solutions with the new formulation. However, the number of decision variables $\{P_{m}(\vec{R})\}_{m=1}^{M}$
in the heterogeneous channel model equals to $M2^{M}$,
which causes exponential blow up in the computational complexity. We next
focus on developing a low complexity efficient heuristic algorithm  to solve the MDP.

Recall that in the heuristic algorithm in Lemma \ref{lemma12s} for the homogeneous
channel recommendation, the weight of selecting each recommended channel
is $\frac{1}{N}$ and total weights of choosing recommended channels
are $R\frac{1}{N}$. Similarly, we can design a low complexity heuristic
algorithm for the heterogeneous channel recommendation. More specifically,
we set the weight of selecting channel $m$ is $P_{1}^{m}$ ($P_{0}^{m}$, respectively)
when the channel is recommended (the channel is not recommended, respectively).
Given the system is in state $\vec{R}$, the probability of choosing
channel $m$ is proportional to its weight of its state $I_{m}$,
i.e.,\begin{equation}
P_{m}(\vec{R})=\frac{P_{I_{m}}^{m}}{\sum_{m'=1}^{M}P_{I_{m'}}^{m}}.\label{eq:hcp}\end{equation}
In this case, the total number of decision variables $P_{I_{m}}^{m}$ is
reduced to $2M$, which grows linearly in the number of channels $M$. Let $\vec{\pi}=\{(P_{1}^{m},P_{0}^{m})\}_{m=1}^{M}\in(0,1)^{2M}$denote
the set of corresponding decision variables. Our objective is to find
the optimal $\vec{\pi}$ that maximizes the time average throughput
$\Phi_{\vec{\pi}}$. We can again apply the MRAS method to find the optimal
solution, which is given in Algorithm \ref{alg:MRAS-Method-ForII}.
The procedures of derivation is very similar with the MRAS method
for the homogeneous channel recommendation; we omit the details due to
space limit. 

Note that the optimal policy $\vec{\pi}^{*}$ for the
heuristic heterogeneous channel recommendation is also a feasible policy for
the heterogeneous channel recommendation MDP. The performance of the
optimal policy for the heterogeneous channel recommendation MDP thus
dominates the heuristic heterogeneous channel recommendation. However,
numerical results show that the heuristic heterogeneous channel recommendation
has a small performance loss comparing to the optimal policy while gaining a significant computation complexity
reduction.
}

\begin{algorithm}[tt]
\begin{algorithmic}[1]
\State \textbf{initialize} parameters for the elite ratio $\rho$, Gaussian distributions $\boldsymbol{\mu}(0)=\{(\mu_{1}^{m}(0),\mu_{0}^{m}(0))\}_{m=1}^{M},\boldsymbol{\sigma}(0)=\{(\sigma_{1}^{m}(0),\sigma_{0}^{m}(0))\}_{m=1}^{M}$,
and the stopping criterion $\xi$. Set initial elite threshold $\gamma_{0}=0$ and iteration index $k=0$.
\Repeat{:}
\State \textbf{increase} iteration index $k$ by 1.
\State \textbf{generate} $L$ candidate policies $\vec{\pi}_{1},...,\vec{\pi}_{L}$
from the random policy generation mechanism $f(\vec{\pi},\boldsymbol{\mu}(k-1),\boldsymbol{\sigma}(k-1))$.
\State \textbf{select} elite policies by setting the elite threshold $\gamma_{k}=\max\{\Phi_{\hat{\vec{\pi}}_{\lceil(1-\rho)L\rceil}},\gamma_{k-1}\}.$
\State \textbf{update} the random policy generation mechanism by (for any $I_{m}\in\{0,1\},m\in\mathcal{M}$) \begin{align}
\mu_{I_{m}}^{m}(k) & =  \frac{\sum_{i=1}^{L}e^{(k-1)\Phi_{\vec{\pi}}}I_{\{\Phi_{\vec{\pi}_{i}}\geq\gamma_{k}\}}P_{I_{m}}^{m}}{\sum_{i=1}^{L}e^{(k-1)\Phi_{\vec{\pi}}}I_{\{\Phi_{\vec{\pi}_{i}}\geq\gamma_{k}\}}},\label{eq:525-1}\\
\sigma_{I_{m}}^{m}(k) & = \left(\frac{\sum_{i=1}^{L}e^{(k-1)\Phi_{\vec{\pi}}}I_{\{\Phi_{\vec{\pi}_{i}}\geq\gamma_{k}\}}(P_{I_{m}}^{m}-\mu_{I_{m}}^{m}(k))^{2}}{\sum_{i=1}^{L}e^{(k-1)\Phi_{\vec{\pi}}}I_{\{\Phi_{\vec{\pi}_{i}}\geq\gamma_{k}\}}}\right)^{\frac{1}{2}}.\label{eq:525-2}\end{align}
\Until{$\max_{I_{m}\in\{0,1\},m\in\mathcal{M}}\sigma_{I_{m}}^{m}(k)<\xi.$}
\end{algorithmic}
\caption{\label{alg:MRAS-Method-ForII}MRAS-based Algorithm For Optimizing Heuristic Heterogeneous Channel Recommendation}
\end{algorithm} 

\section{\label{sec:Numerical-Results}Simulation Results}

In this section, we investigate the proposed adaptive channel recommendation scheme by simulations. The results show that the adaptive channel recommendation scheme
not only achieves a higher performance over the static channel recommendation
scheme and random access scheme, but also is more robust to the dynamic
change of the channel environments.

\subsection{Simulation Setup}

We first consider a cognitive radio network consisting of multiple independent
and stochastically homogeneous primary channels.  The data rate of each channel is normalized to be $1$ Mbps. In order to take
the impact of primary user's long run behavior into account, we consider the following two types of channel state transition matrices: \begin{eqnarray}
\mbox{Type 1: } \Gamma^{1} & = & \left[\begin{array}{cc}
1-0.005\epsilon & 0.005\epsilon\\
0.025\epsilon & 1-0.025\epsilon\end{array}\right],\label{T1}\\
\mbox{Type 2: } \Gamma^{2} & = & \left[\begin{array}{cc}
1-0.01\epsilon & 0.01\epsilon\\
0.01\epsilon & 1-0.01\epsilon\end{array}\right],\label{T2}\end{eqnarray}
where $\epsilon$ is the dynamic factor. Recall that a larger $\epsilon$ means that the channels are more dynamic over time. Using (\ref{eq:sd-2}), we know that channel models $\Gamma^{1}$ and $\Gamma^{2}$ have the stationary channel idle probabilities of ${1}/{6}$
and ${1}/{2}$, respectively. In other words, the primary activity level is much higher with the Type 1 channel than with the Type 2 channel.

We initialize the parameters of MRAS algorithm as follows. We set $\mu_{R}=0.5$ and $\sigma_{R}=0.5$ for the Gaussian distribution, which has 68.2\% support over the feasible region $(0,1)$. We found that the performance of the MRAS algorithm is insensitive to the elite ratio $\rho$ when $\rho\leq0.3$. We thus choose  $\rho=0.1$.

When using the MRAS-based algorithm, we need to determine how many (feasible) candidate policies to generate in each iteration.  Figure \ref{fig:MRAS-algorithm-with}
shows the convergence of MRAS algorithm with $100$, $300$, and $500$ candidate policies per iteration, respectively. We have two observations. First,
the number of iterations to achieve convergence reduces as the number
of candidate policies increases. Second, the convergence speed is insignificant when the number changes from $300$ to $500$. We thus choose
$L=500$ for the experiments in the sequel.

\begin{figure}[ht]
\begin{center}
\includegraphics[scale=0.45]{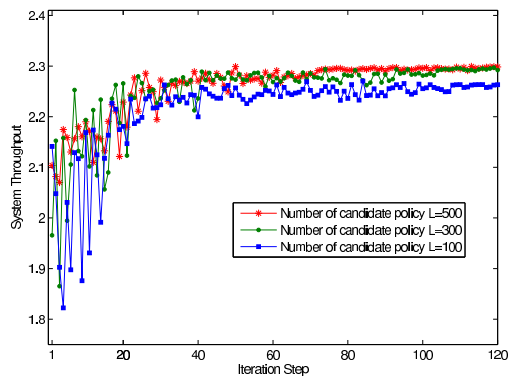}
\caption{\label{fig:MRAS-algorithm-with}The convergence of MRAS-based algorithm with different number
of candidate policies per iteration}
\end{center}
\end{figure}

\subsection{Simulation Results}
We implement the adaptive channel recommendation
scheme with  $M=10$ channels and $N=5$ secondary users. We also benchmark the adaptive channel recommendation
scheme with the static channel recommendation scheme in \cite{key-2} and the random access scheme as the benchmark. We choose the dynamic factor $\epsilon$ within a wide range to investigate the
robustness of the schemes to the channel dynamics. The results are shown in Figures \ref{fig:Type-1-Transition1} -- \ref{fig:Performance_Gain2}. From these figures, we see that
\begin{itemize}
\item \emph{Superior performance of adaptive channel recommendation scheme (Figures \ref{fig:Type-1-Transition1} and \ref{fig:Type-2-Transition1})}: the adaptive channel recommendation scheme performs better than the random
access scheme and static channel recommendation scheme. Typically,
it offers 5\%$\scriptsize{\sim}$18\% performance gain over the static
channel recommendation scheme.
\item  \emph{Impact of channel dynamics (Figures \ref{fig:Type-1-Transition1} and \ref{fig:Type-2-Transition1})}: the performances of both adaptive and static channel recommendation schemes degrade as the dynamic factor
$\epsilon$ increases. The reason is that both two schemes rely on
the recommendation information from previous time slots to make decisions.
When channel states change rapidly, the value of recommendation information
diminishes. However, the adaptive channel recommendation is much more robust to the dynamic channel environment changing (See Figure \ref{fig:Performance_Gain2}). This
is because the optimal adaptive policy takes the channel dynamics into account while the static one does not.
\item \emph{Impact of channel idleness level (Figures \ref{fig:Performance_Gain} and \ref{fig:Performance_Gain2})}: Figure \ref{fig:Performance_Gain} shows the performance gain of  the adaptive channel recommendation scheme over the random access scheme under two different types of transition matrix scenarios. We see that the performance gain decreases with the idle probability of the channel. This shows that the information of channel recommendations can enhance the spectrum access more efficiently when the primary activity level increases (i.e., when the channel idle probability is low). Interestingly,  Figure \ref{fig:Performance_Gain2} shows that the performance gain of the adaptive channel recommendation scheme over the static channel recommendation scheme trends to increase with the channel idleness probability. This illustrates that the adaptive channel recommendation scheme can better utilize the channel opportunities given the information of channel recommendations.
\end{itemize}



\begin{figure}[ht]
\begin{center}
\includegraphics[scale=0.45]{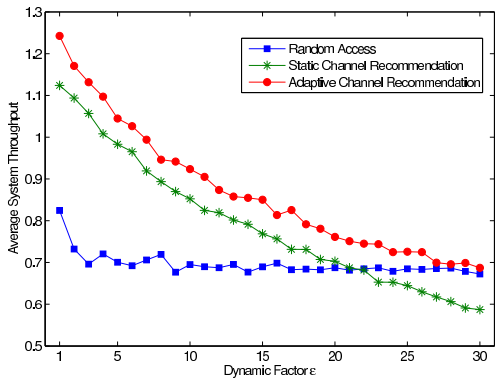}
\caption{\label{fig:Type-1-Transition1}System throughput with $M=10$ channels and $N=5$ users under the Type 1 channel state transition matrix}
\end{center}
\end{figure}

\begin{figure}[ht]
\begin{center}
\includegraphics[scale=0.45]{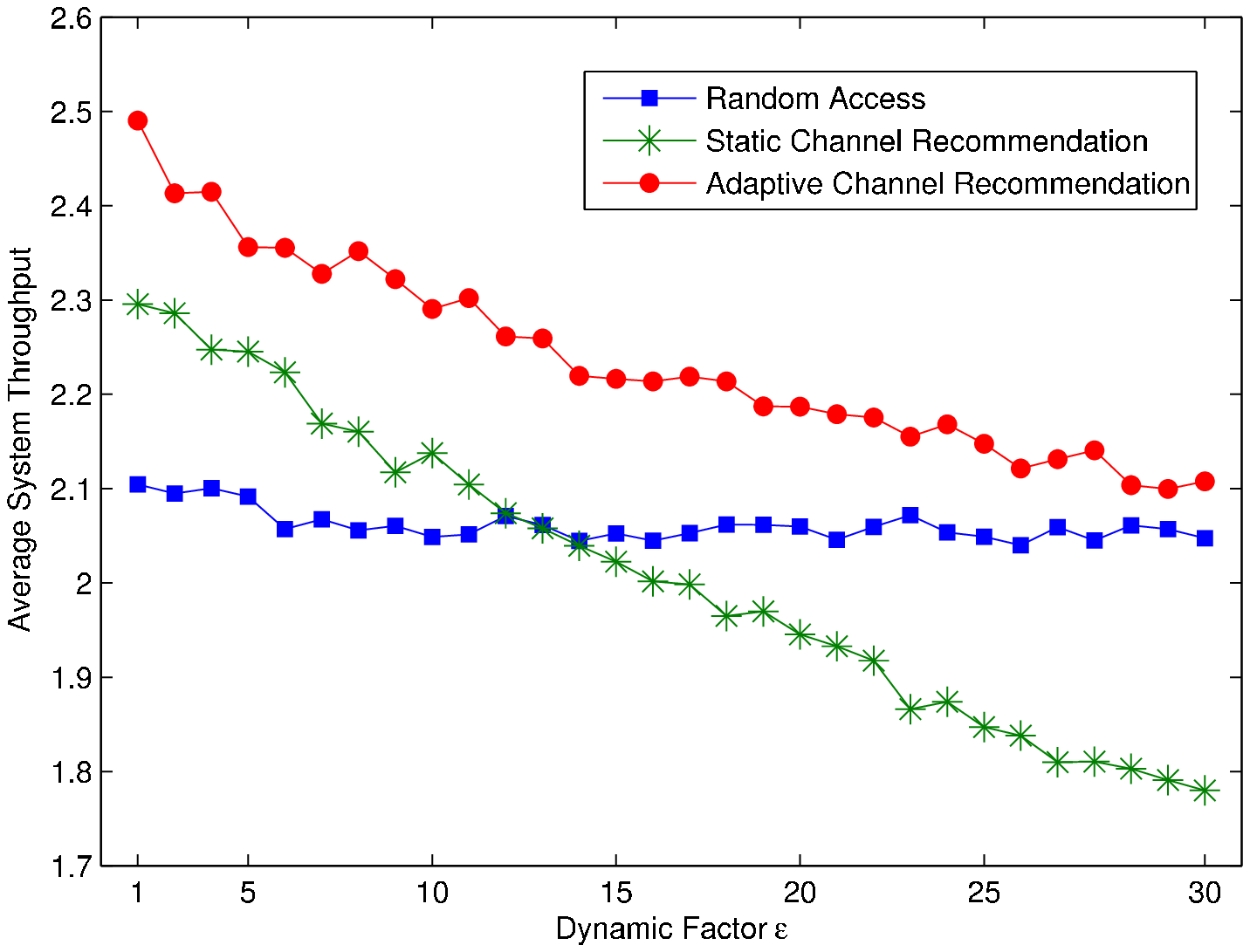}
\caption{\label{fig:Type-2-Transition1}System throughput with $M=10$ channels and $N=5$ users  under the Type 2 channel state transition matrix}
\end{center}
\end{figure}

\begin{figure}[ht]
\begin{center}
\includegraphics[scale=0.45]{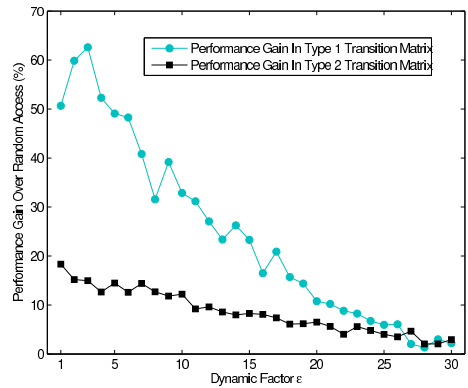}
\caption{\label{fig:Performance_Gain}Performance gain over random access scheme. The Type 1 and Type 2 channels have the stationary channel idle probabilities of ${1}/{6}$ and ${1}/{2}$, respectively.}
\end{center}
\end{figure}

\begin{figure}[ht]
\begin{center}
\includegraphics[scale=0.45]{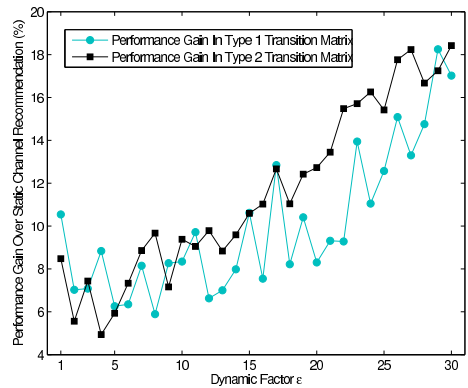}
\caption{\label{fig:Performance_Gain2}Performance gain over static channel recommendation scheme. The Type 1 and Type 2 channels have the stationary channel idle probabilities of ${1}/{6}$ and ${1}/{2}$, respectively.}
\end{center}
\end{figure}

%
%

\subsection{Comparison of MRAS algorithm and Q-Learning}

To benchmark the performance of the spectrum access
policy based on the MRAS algorithm, we compare it
with the policy obtained by Q-learning algorithm \cite{key-22}.

Since the Q-learning can only be used over the discrete action space, we
first discretize the action space $\mathcal{P}$ into
a finite discrete action space $\mathcal{\hat{P}}=\{0.1,...,1.0\}$.
The Q-learning then defines a Q-value representing the estimated quality
of a state-action combination as $
Q:\mathcal{R}\times\mathcal{\hat{P}}_{rec}\rightarrow\mathbb{R}.$
Given a new reward $U(R(t),P_{rec}(t))$ is received, we can update the Q-value to be \begin{multline*}
Q(R(t),P_{rec}(t))= (1-\alpha)Q(R(t),P_{rec}(t)),\\
 +\alpha[U(R(t),P_{rec}(t))+\max_{P_{rec}\in\mathcal{\hat{P}}}Q(R(t+1),P_{rec})],
 \end{multline*}
where $0<\alpha<1$ is the smoothing factor. Given a system
state $R$, the probability of choosing an action $P_{rec}$ is $
P_{r}(P_{rec}(t)=P_{rec}|R(t)=R)=\frac{e^{\tau Q(R,P_{rec})}}{\sum_{P_{rec}^{'}\in\mathcal{\hat{P}}}e^{\tau Q(R,P_{rec})}},$
where $\tau>0$ is the temperature.

After the Q-learning converges, we obtain the corresponding spectrum
access policy $\pi_{Q}$ over the discretized action space $\mathcal{\hat{P}}$.
Note that $\pi_{Q}$ is a sub-optimal policy for the adaptive channel
recommendation MDP over the continuous action space $\mathcal{P}$.

We compare the Q-learning based policy with our MRAS-based optimal policy when there are $M=10$ channels and
$N=5$ users, and show the simulation results in Figures \ref{fig:Type-1-transition3} and \ref{fig:Type-2-transition3}.
From these figures, we see that the MRAS-based
algorithm outperforms Q-learning up to $10\%$, which demonstrates the effectiveness of our proposed algorithm.

\begin{figure}[ht]
\begin{center}
\includegraphics[scale=0.45]{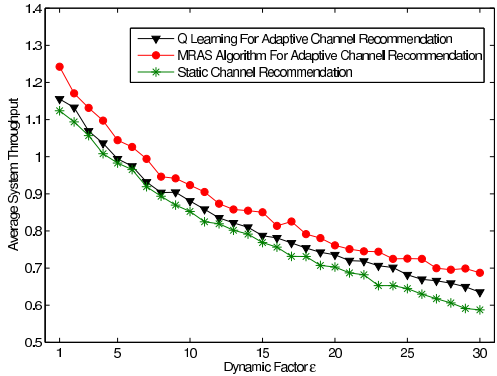}
\caption{\label{fig:Type-1-transition3}Comparison of MRAS-based algorithm and Q-learning with Type 1 channel state transition matrix}
\end{center}
\end{figure}

\begin{figure}[ht]
\begin{center}
\includegraphics[scale=0.45]{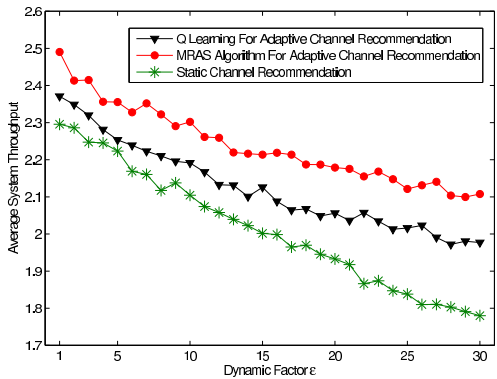}
\caption{\label{fig:Type-2-transition3}Comparison of MRAS-based algorithm and Q-learning with Type 2 channel state transition matrix}
\end{center}
\end{figure}

\subsection{Heuristic Heterogenous Channel Recommendation}
\rev{
We now evaluate the proposed heuristic heterogeneous channel recommendation mechanism in Section \ref{sec:Adaptive-Channel-RecommendationII} with a network consisting of $M=10$ channels and $N=5$ users. We implement the heuristic heterogeneous channel recommendation mechanism in both homogeneous and heterogenous homogeneous environments.
\subsubsection{Homogeneous Channel Environment}
We first study how the heuristic heterogeneous channel recommendation mechanism performs in the homogeneous channel environment (which is a special case of the heterogeneous environment) in both types of $\Gamma^{1}$ and $\Gamma^{2}$ homogeneous channel environments, and simulate the optimal homogeneous channel recommendation (Algorithm \ref{alg:MRAS-Method-For}) as a benchmark. . The data rate of each channel is normalized to be $1$ Mbps. The results are shown in Figures \ref{fig:Type-1-transition3II} and \ref{fig:Type-2-transition3II}. Comparing to the optimal channel access policy, the performance loss of the heuristic heterogeneous channel recommendation in the Type $1$ and Type $2$ channel environments are at most $12\%$ and $5\%$, respectively. This shows the efficiency of the heuristic heterogeneous channel recommendation in homogeneous channel environments.

\begin{figure}[ht]
\begin{center}
\includegraphics[scale=0.45]{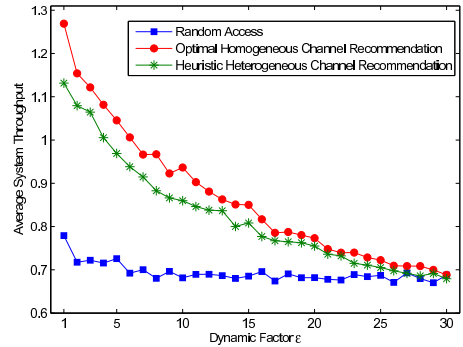}
\caption{\label{fig:Type-1-transition3II}Comparison of heuristic heterogenous channel recommendation and optimal homogeneous channel recommendation in Type 1 homogeneous channel environment.}
\end{center}
\end{figure}

\begin{figure}[ht]
\begin{center}
\includegraphics[scale=0.45]{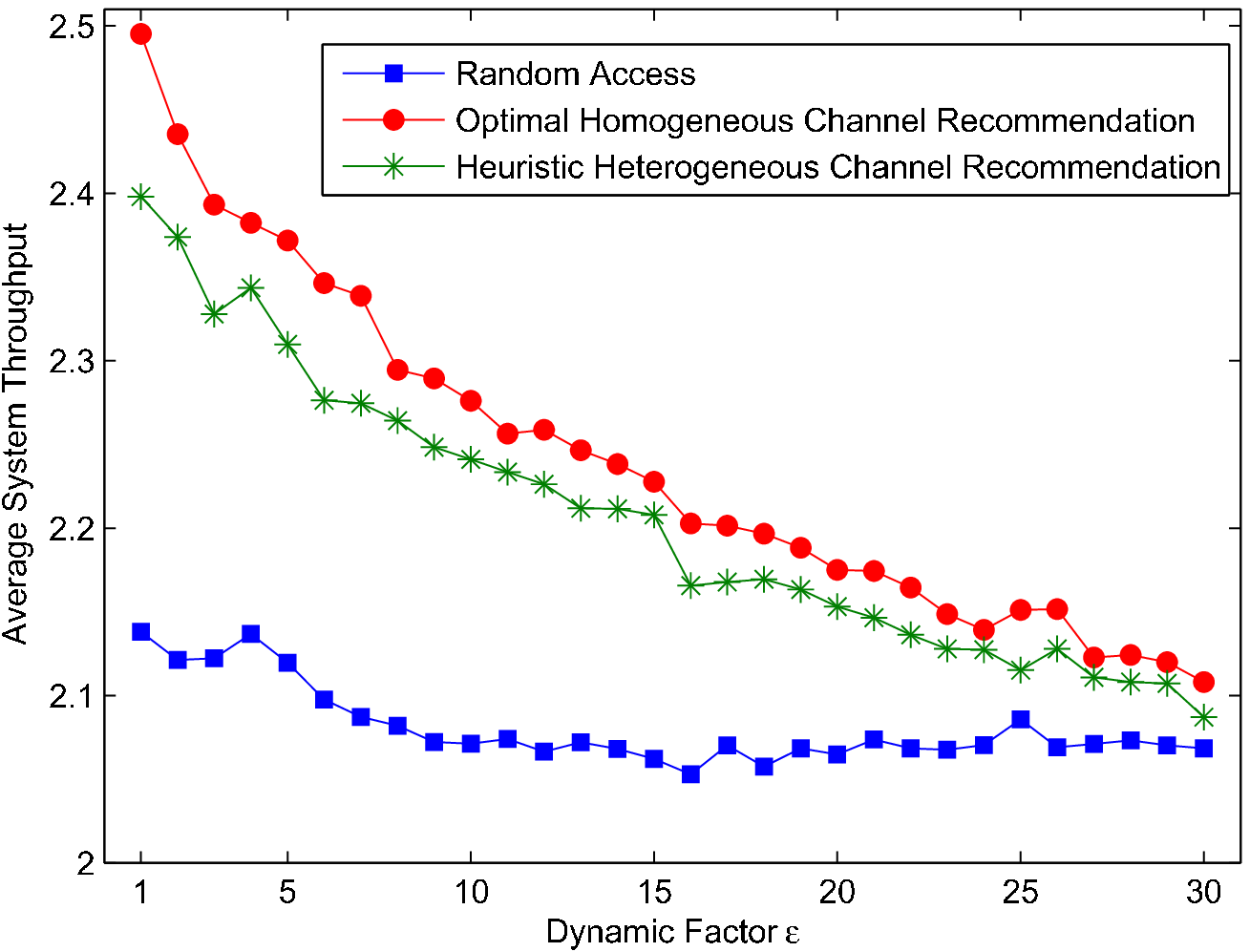}
\caption{\label{fig:Type-2-transition3II}Comparison of heuristic heterogenous channel recommendation and optimal homogeneous channel recommendation in Type 2 homogeneous channel environment.}
\end{center}
\end{figure}

\begin{figure}[ht]
\begin{center}
\includegraphics[scale=0.45]{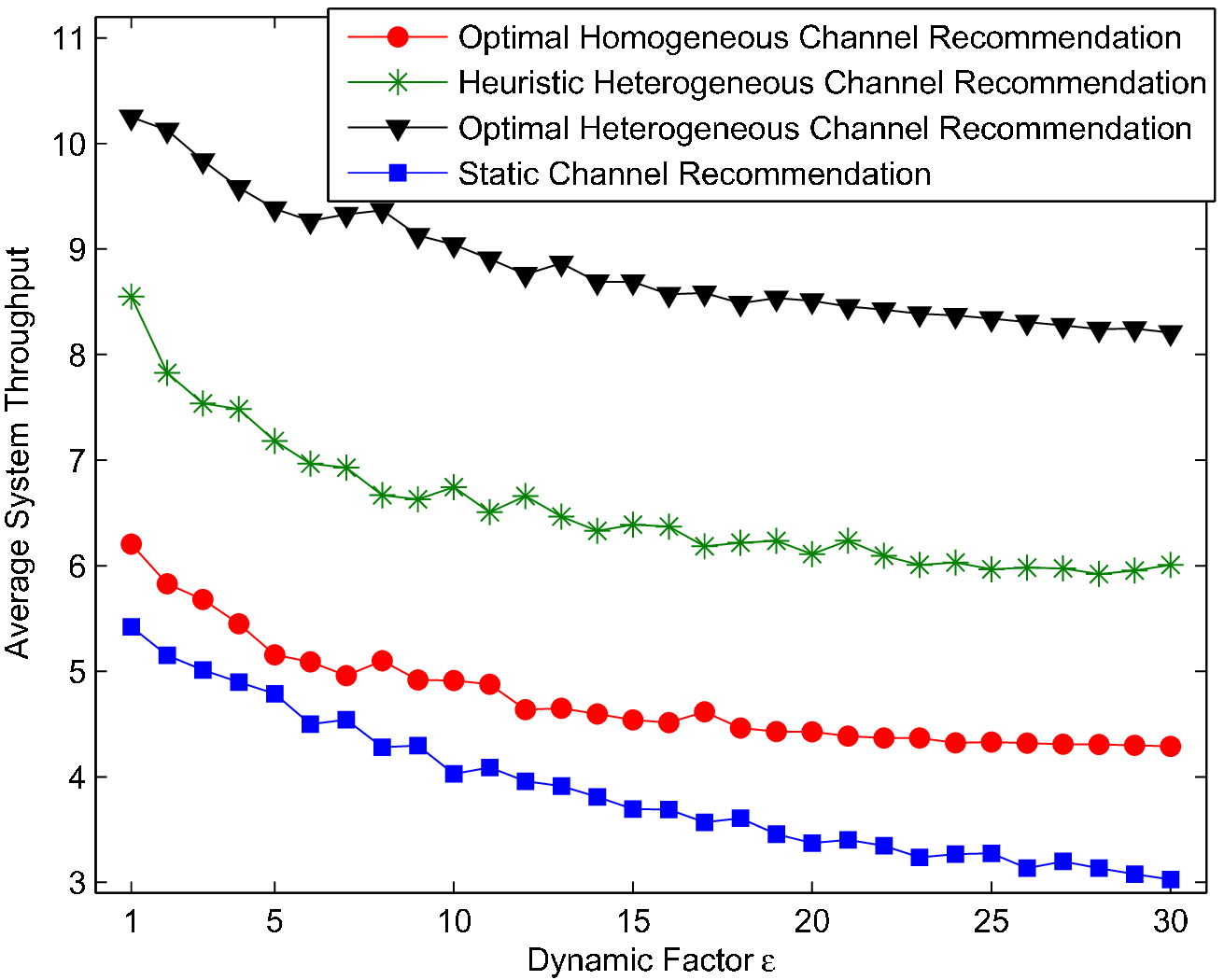}
\caption{\label{fig:Type-1-transition3III}Comparison of heuristic heterogenous channel recommendation, optimal homogeneous channel recommendation and optimal homogeneous channel recommendation in the first kind of heterogenous channel environment.}
\end{center}
\end{figure}

\begin{figure}[ht]
\begin{center}
\includegraphics[scale=0.45]{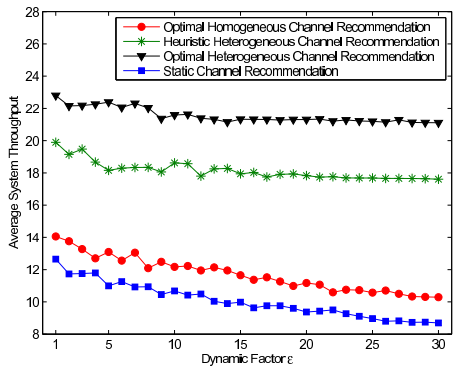}
\caption{\label{fig:Type-2-transition3III}Comparison of heuristic heterogenous channel recommendation, optimal homogeneous channel recommendation and optimal homogeneous channel recommendation in the second kind of heterogenous channel environment.}
\end{center}
\end{figure}

\subsubsection{Heterogeneous Channel Environment}
We next implement the heuristic heterogeneous channel recommendation mechanism in heterogenous channel environments. The data rates of $M=10$ channels are $\{B_{1}=0.2,B_{2}=0.6,B_{3}=0.8,B_{4}=1,B_{5}=2,B_{6}=4,B_{7}=6,B_{8}=8,B_{9}=10,B_{10}=20\}$ Mbps. We consider two kinds of stochastic channel state changing environments:
\begin{align}&\{\Gamma_{1}=\Gamma^{2},\Gamma_{2}=\Gamma^{2},\Gamma_{3}=\Gamma^{2},\Gamma_{4}=\Gamma^{2},\Gamma_{5}=\Gamma^{2},\nonumber \\
&~\Gamma_{6}=\Gamma^{1},\Gamma_{7}=\Gamma^{1},\Gamma_{8}=\Gamma^{1},\Gamma_{9}=\Gamma^{1},\Gamma_{10}=\Gamma^{1}\},\end{align}
and\begin{align}
&\{\Gamma_{1}=\Gamma^{1},\Gamma_{2}=\Gamma^{1},\Gamma_{3}=\Gamma^{1},\Gamma_{4}=\Gamma^{1},\Gamma_{5}=\Gamma^{1},\nonumber \\ &~\Gamma_{6}=\Gamma^{2},\Gamma_{7}=\Gamma^{2},\Gamma_{8}=\Gamma^{2},\Gamma_{9}=\Gamma^{2},\Gamma_{10}=\Gamma^{2}\}.\end{align} Here subscript denotes channel index, and superscript denote channel type index. For the first kind of channel environment, a channel with low data rate tends to have a low primary transmission occupancy. While for the second kind, a channel with high data rate tends to have a high idleness probability. We also implement static channel recommendation, the optimal homogeneous channel recommendation (Algorithm \ref{alg:MRAS-Method-For}) and optimal heterogeneous channel recommendation (obtained by adapting the MRAS algorithm to optimize the heterogeneous channel MDP, not shown in this paper) as benchmarks. The results are depicted in Figures \ref{fig:Type-1-transition3III} and \ref{fig:Type-2-transition3III}. From these figures, we see that:
\begin{itemize}
\item For the first kind of channel environment, the heuristic heterogeneous channel recommendation achieves up-to
$40\%$ and $100\%$ performance improvement over the optimal homogeneous channel
recommendation and the static channel recommendation, respectively. Comparing with
the optimal heterogeneous channel recommendation, the performance
loss of the heuristic heterogeneous channel recommendation is at most
$35\%$. Note that the number of decision variables in the optimal heterogeneous
channel recommendation is $M2^{M}=10240$, while the number of decision variables
in the heuristic heterogeneous channel recommendation is only $2M=20$.
The convergence of the heuristic heterogeneous channel recommendation
hence is much faster than the optimal heterogeneous channel recommendation.
\item For the second kind of channel environment, the heuristic heterogeneous
channel recommendation achieves up-to $70\%$ and $100\%$ performance improvement
over the optimal homogeneous channel recommendation and static channel recommendation, respectively. The performance
loss is at most
$20\%$ comparing with the the optimal heterogeneous channel recommendation. Comparing with Figure \ref{fig:Type-1-transition3III}, we see that  the heuristic heterogeneous channel recommendation
performs better if more channel opportunities are available for the
secondary users.
\end{itemize}
} 

\section{\label{sec:Related-Work}Related Work}

The spectrum access by multiple secondary users can be either \emph{uncoordinated} or \emph{coordinated}.
For the uncoordinated case, multiple secondary users compete with other for the resource. Huang \emph{et al.} in \cite{ key-28 } designed two auction mechanisms to allocate the interference budget among selfish users. Southwell and Huang in \cite{key-88} studied the largest and smallest convergence time to an equilibrium when secondary users access multiple channels in a distributed fashion.
 Liu \emph{et al.} in  \cite{key-16} modeled the interactions among spatially separated users as congestion games with resource reuse.  Li and Han in \cite{key-29} applied the graphic game theory to address the spectrum access problem with limited range of mutual interference. Anandkumar \emph{et al.} in \cite{key-25} proposed a learning-based approach for competitive spectrum access with incomplete spectrum information. Law \emph{et al.} in \cite{key-15} showed that uncoordinated spectrum access may lead to poor system performance.

For the coordinated spectrum access, Zhao \emph{et al.} in  \cite{key-26} proposed a dynamic group formation algorithm to distribute secondary users' transmissions across multiple channels. Shu and Krunz proposed a multi-level spectrum opportunity framework in \cite{key-27}. The above papers assumed that each secondary user knows the entire channel occupancy information. We consider the case where  each secondary user only has a limited view of the system, and improve each other's information by recommendation.


\rev{Our algorithm design is partially inspired by the recommendation systems in the electronic commerce industry, where analytical methods such as collaborative filtering \cite{key-10} and
multi-armed bandit process modeling \cite{key-14} are useful. However, we cannot directly apply the existing methods to analyze cognitive radio networks due to the unique congestion effect in our model.}


\section{\label{sec:Conclusion}Conclusion}
In this paper, we propose an adaptive channel recommendation scheme for efficient spectrum sharing. We formulate the problem as an average reward based Markov decision process. We first prove the existence of the optimal stationary spectrum access policy, and then characterize the structure of the optimal policy in two asymptotic cases. Furthermore, we propose a novel MRAS-based algorithm that is provably convergent to the optimal policy.  Numerical results show that our proposed algorithm outperforms the static approach in the literature by up to $18\%$ and the Q-learning method by up to $10\%$ in terms of system throughput. Our algorithm is also more robust to the channel dynamics compared to the static counterpart.

In terms of future work, we are currently extending the analysis by taking the heterogeneity of channels into consideration. We also plan to consider the case where the secondary users are selfish. Designing an incentive-compatible channel recommendation mechanism for that case will be very interesting and challenging. 

\appendix
\subsection{Proof of Lemma \ref{lemma11s}}\label{ProofLemma1}
When $S_{m}(t)=0$, this trivially holds. We focus on the case that
$S_{m}(t)=1$.

Let $\mathcal{K}_{m}=\{1,...,k_{m}(t)\}$ be the set of secondary
users accessing the channel $m$, $\tau_{m}^{i}$ be the backoff time
be generated by secondary user $i$ and $\tau_{m}^{(1)}=\min\{\tau_{m}^{i}|i\neq n,i\in\mathcal{K}_{m}\}$.
The probability that the user $n$ captures the channel $m$ is given
as\begin{eqnarray*}
Pr_{n,m} & = & P\{\tau_{m}^{(1)}>\tau_{m}^{n}\}\\
 & = & (1-\frac{\tau_{m}^{n}}{\tau_{max}})^{k_{m}(t)-1}.\end{eqnarray*}
Thus, the expected throughput of user $n$ is\begin{eqnarray*}
u_{n}(t) & = & \int_{0}^{\tau_{max}}BPr_{n,m}\frac{1}{\tau_{max}}d\tau_{m}^{n}\\
 & = & \int_{0}^{\tau_{max}}B(1-\frac{\tau_{m}^{n}}{\tau_{max}})^{k_{m}(t)-1}\frac{1}{\tau_{max}}d\tau_{m}^{n}\\
 & = & \frac{B}{k_{m}(t)}.\\
 & = & \frac{BS_{m}(t)}{k_{m}(t)}.\end{eqnarray*} \qed

\subsection{Proof of Lemma \ref{lemma12s}}\label{ProofLemma2}
Let $\Lambda_{C}$ denote the event that $C$ secondary users choose
the recommended channels, and $Pr(c_{1},...,c_{R})$ denote probability
mass function that the number of secondary users on these $R$ recommended
channels equal to $c_{1},...,c_{R}$ respectively. Given the event
$\Lambda_{C}$, we have \[
Pr(c_{1},...,c_{R}|\mbox{\ensuremath{\Lambda_{C}}})=\left(\begin{array}{c}
n\\
c_{1},...,c_{R}\end{array}\right)R^{-C},\]
which is a Multinomial mass function. By the property of Multinomial
distribution, we have\[
E[c_{m}|\Lambda_{C}]=\frac{C}{R}.\]
It follows that the expected number of users choosing a recommended
channel $m$ is\begin{eqnarray*}
E[c_{m}] & = & \sum_{C=0}^{N}E[c_{m}|\Lambda_{C}]Pr(\Lambda_{C})\\
 & = & \sum_{C=0}^{N}\frac{C}{R}\left(\begin{array}{c}
N\\
C\end{array}\right)P_{rec}^{C}(1-P_{rec})^{N-C}\\
 & = & \frac{P_{rec}N}{R}.\end{eqnarray*}
Then $E[c_{m}]=1$ requires that\[
P_{rec}=\frac{R}{N}.\]  \qed

\subsection{Derivation of Transition Probability}\label{Derivation}
When the system state transits from $R$ to $R'$, we assume that
$m_{r}$ and $m_{u}$ recommendations, out of $R'$ recommendations,
are channels that have been recommended and have not been recommended
at time slot $t$ respectively. Obviously, $m_{r}+m_{u}=R'$. We assume
that $\bar{m}_{r}$ recommended channels and $\bar{m}_{u}$ unrecommended
channels have been accessed by the secondary users at time slot $t+1$.
We thus have $R\geq\bar{m}_{r}\geq m_{r}$ and $M-R\geq\bar{m}_{u}\geq m_{u}$.
We also assume that there are $n_{r}$ secondary users have accessed
these $\bar{m}_{r}$ recommended channels and $n_{u}$ secondary users
have accessed those $\bar{m}_{u}$ unrecommended channels at time
slot $t+1$. Obviously, we have $n_{r}+n_{u}=N$ , $n_{r}\geq\bar{m}_{r}$
and $n_{u}\geq\bar{m}_{u}$.

For the first term, the probability that the user distribution $(n_{r},n_{u})$
happens follows the Binomial distribution as $\left(\begin{array}{c}
N\\
n_{r}\end{array}\right)P_{rec}^{n_{r}}(1-P_{rec})^{n_{u}}$.

For the second term, when $\bar{m}_{r}\ge1$, it is easy to check
that there are $\left(\begin{array}{c}
n_{r}-1\\
\bar{m}_{r}-1\end{array}\right)$ ways for $n_{r}$ secondary users to choose $\bar{m}_{r}$ recommended
channels and there are $\frac{R!}{(R-\bar{m}_{r})!}$ possibilities
for these $\bar{m}_{r}$ recommended channels out of the $R$ recommended
channels, each of which has probability $(\frac{1}{R})^{n_{r}}$.
Among these $\bar{m}_{r}$ recommended channels that have been accessed
by the secondary users, the probability that $m_{r}$ channels turn
out to be idle is given as $\left(\begin{array}{c}
\bar{m}_{r}\\
m_{r}\end{array}\right)(1-q)^{m_{r}}q^{\bar{m}_{r}-m_{r}}$. When $\bar{m}_{r}=0$, it requires that $u_{r}=0$. Thus, we define
\[
\left(\begin{array}{c}
n_{r}-1\\
-1\end{array}\right)=\begin{cases}
1 & \mbox{If \ensuremath{n_{r}}=0,}\\
0 & \mbox{Otherwise.}\end{cases}\]

Similarly, we can obtain the third term for the unrecommended channels
case.

\subsection{Lemma 5}\label{proof_for_lemma1}
Since the operation $\sum_{R'\in\mathcal{R}}P{}_{R,R'}^{P_{rec}}[\cdot]$
plays a key role in the Bellman equation, to facilitate the study,
we first define the following function \begin{eqnarray*}
f_{r}(R,P_{rec}) & \triangleq & \sum_{i=r}^{\min\{M,N\}}P{}_{R,i}^{P_{rec}},\forall r\in\mathcal{R}.\end{eqnarray*}
Since\begin{eqnarray*}
 &  & f_{r}(R,P_{rec})\\
 &  & =Pr(R(t+1)\geq r|R(t)=R,P_{rec}(t)=P_{rec})\\
 &  & =1-Pr(R(t+1)<r|R(t)=R,P_{rec}(t)=P_{rec}),\end{eqnarray*}
We call the function $f_{r}(R,P_{rec})$ as \emph{the reverse cumulative distribution
function} in the sequel.

\begin{lem}
\label{lem:nondecreaseing} When $M=\infty$ and $N<\infty$,
the reverse cumulative distribution function $f_{r}(R,P_{rec})$ is
nondecreasing in $R$ for all $r,R\in\mathcal{R}$, $P_{rec}\in\mathcal{P}$.\end{lem}

\emph{proof:} We prove the result by induction argument. In abuse of notation, we
denote the transition probability $P{}_{R,R'}^{P_{rec}}$ and the
reverse cumulative distribution function $f_{r}(R,P_{rec})$ when
the number of users $N=k$ as $P_{R,R'}^{P_{rec}}(k)$ and $f_{r}^{k}(R,P_{rec})$
respectively.

When $N=2$, from (\ref{eq:2565}), we have\begin{eqnarray*}
P_{0,0}^{P_{rec}}(2) & = & P_{rec}^{2}+(1-P_{rec})^{2}(\frac{q}{p+q})^{2}\\
 &  & + 2P_{rec}(1-P_{rec})\frac{q}{p+q},\end{eqnarray*}
\[
P_{0,1}^{P_{rec}}(2)=(1-P_{rec})^{2}\frac{2pq}{(p+q)^{2}}+2P_{rec}(1-P_{rec})\frac{p}{p+q},\]
\[
P_{0,2}^{P_{rec}}(2)=(1-P_{rec})^{2}(\frac{p}{p+q})^{2},\]
\begin{eqnarray*}
P_{1,0}^{P_{rec}}(2) & = & P_{rec}^{2}q+(1-P_{rec})^{2}(\frac{q}{p+q})^{2}\\
 & & + 2P_{rec}(1-P_{rec})\frac{q^{2}}{p+q},\end{eqnarray*}
\begin{eqnarray*}
P_{1,1}^{P_{rec}}(2) & = & P_{rec}^{2}(1-q)+(1-P_{rec})^{2}\frac{2pq}{(p+q)^{2}}\\
 & & + 2P_{rec}(1-P_{rec})\frac{(1-q)q+pq}{p+q},\end{eqnarray*}
\[
P_{1,2}^{P_{rec}}(2)=(1-P_{rec})^{2}(\frac{p}{p+q})^{2}+2P_{rec}(1-P_{rec})\frac{(1-q)p}{p+q},\]
\begin{eqnarray*}
P_{2,0}^{P_{rec}}(2) & = & P_{rec}^{2}\frac{q+q^{2}}{2}+(1-P_{rec})^{2}(\frac{q}{p+q})^{2}\\
 & & + 2P_{rec}(1-P_{rec})\frac{q^{2}}{p+q},\end{eqnarray*}
\begin{eqnarray*}
P_{2,1}^{P_{rec}}(2) & = & P_{rec}^{2}\frac{1-q+(1-q)q}{2}+(1-P_{rec})^{2}\frac{2pq}{(p+q)^{2}}\\
 & & + 2P_{rec}(1-P_{rec})\frac{(1-q)q+pq}{p+q},\end{eqnarray*}
\begin{eqnarray*}
P_{2,2}^{P_{rec}}(2) & = & P_{rec}^{2}\frac{(1-q)^{2}}{2}+(1-P_{rec})^{2}(\frac{p}{p+q})^{2}\\
 & & + 2P_{rec}(1-P_{rec})\frac{(1-q)p}{p+q}.\end{eqnarray*}
It is easy to check the following holds\begin{eqnarray*}
P_{0,0}^{P_{rec}}(2) & \geq & P_{1,0}^{P_{rec}}(2)\geq P_{2,0}^{P_{rec}}(2),\\
P_{2,2}^{P_{rec}}(2) & \geq & P_{1,2}^{P_{rec}}(2)\geq P_{0,2}^{P_{rec}}(2).\end{eqnarray*}
Since \[
\sum_{i=0}^{2}P_{0,i}^{P_{rec}}(2)=\sum_{i=0}^{2}P_{1,i}^{P_{rec}}(2)=\sum_{i=0}^{2}P_{2,i}^{P_{rec}}(2)=1,\]
we thus obtain\[
f_{r}^{2}(R+1,P_{rec})\geq f_{r}^{2}(R,P_{rec}),\forall R,r\in\mathcal{R},P_{rec}\in\mathcal{P},\]
i.e. $f_{r}(R,P_{rec})$ is nondecreasing in $R$ for the case $N=2$.

We then assume that $f_{r}(R,P_{rec})$ is nondecreasing in $R$ for
all $R\in\mathcal{R}$, $P_{rec}\in\mathcal{P}$ for the case that
$N=k\geq2$ i.e.\[
f_{r}^{k}(R+1,P_{rec})\geq f_{r}^{k}(R,P_{rec}),\forall R,r\in\mathcal{R},P_{rec}\in\mathcal{P}.\]
We next prove that $f_{r}(R,P_{rec})$ is nondecreasing for the case
the $N=k+1$ under this hypothesis.

Let $\psi$ denote the event that one arbitrary user out of these
$k+1$ users, does not generate a recommendation at time slot $t+1$.
Obviously,\[
Pr(\psi)=P_{rec}q+(1-P_{rec})\frac{q}{p+q},\]
which depends on $P_{rec}$ and the channel environment only. By conditioning
on the event $\varphi$, we have\begin{eqnarray}
P{}_{R+1,i}^{P_{rec}}(k+1) & = & P{}_{R+1,i-1}^{P_{rec}}(k)[1-Pr(\psi)]\nonumber \\
 & & + P{}_{R+1,i}^{P_{rec}}(k)Pr(\psi),\label{eq:259}\\
P{}_{R,i}^{P_{rec}}(k+1) & = & P{}_{R,i-1}^{P_{rec}}(k)[1-Pr(\psi)].\nonumber \\
 & & + P{}_{R,i}^{P_{rec}}(k)Pr(\psi)\label{eq:259-1}\end{eqnarray}
Thus, \begin{eqnarray}
 &  & f_{r}^{k+1}(R+1,P_{rec})-f_{r}^{k+1}(R,P_{rec})\nonumber \\
 & = & \sum_{i=r}^{k+1}P{}_{R+1,i}^{P_{rec}}(k+1)-\sum_{i=r}^{k+1}P{}_{R,i}^{P_{rec}}(k+1)\nonumber \\
 & = & [\sum_{i=r}^{k+1}P{}_{R+1,i-1}^{P_{rec}}(k)-\sum_{i=r}^{k+1}P{}_{R,i-1}^{P_{rec}}(k)][1-Pr(\psi)]\nonumber \\
 & & + [\sum_{i=r}^{k}P{}_{R+1,i}^{P_{rec}}(k)-\sum_{i=r}^{k}P{}_{R,i}^{P_{rec}}(k)]Pr(\psi)\nonumber \\
 & = & [\sum_{j=r}^{k}P{}_{R+1,j}^{P_{rec}}(k)-\sum_{j=r}^{k}P{}_{R,j}^{P_{rec}}(k)][1-Pr(\psi)]\nonumber \\
 & & + [\sum_{i=r}^{k}P{}_{R+1,i}^{P_{rec}}(k)-\sum_{i=r}^{k}P{}_{R,i}^{P_{rec}}(k)]Pr(\psi)\nonumber \\
 & = & [f_{r-1}^{k}(R+1,P_{rec})-f_{r-1}^{k}(R,P_{rec})][1-Pr(\psi)]\nonumber \\
 & = & [f_{r}^{k}(R+1,P_{rec})-f_{r}^{k}(R,P_{rec})]Pr(\psi)\nonumber \\
 & \geq & 0.\label{eq:259-2}\end{eqnarray}
i.e. $f_{r}(R,P_{rec})$ is also nondecreasing for the case the $N=k+1$.
By the induction argument, the result holds for the case that $N\geq2$. \qed

\subsection{Lemma 6}\label{proof_for_lemma22}

\begin{lem}
\label{lem:When-,-i.e.}When $M=+\infty$ and $N<+\infty$,
the reverse cumulative distribution function $f_{r}(R,P_{rec})$ is
supermodular on $\mathcal{R}\times\mathcal{P}$.\end{lem}

\emph{proof:} To show $f_{r}(R,P_{rec})$ is supermodular on $\mathcal{R}\times\mathcal{P}$
is equivalent to proving the following is true:\begin{equation}
\frac{\partial^{2}f_{r}(R,P_{rec})}{\partial P_{rec}\partial R}\ge0.\label{eq:2588}\end{equation}
Since $R$ is an integral variable, (\ref{eq:2588}) is equivalent
to \[
\frac{\partial f_{r}(R+1,P_{rec})}{\partial P_{rec}}-\frac{\partial f_{r}(R,P_{rec})}{\partial P_{rec}}\geq0.\]
That is, it is equivalent to showing $\frac{\partial f_{r}(R,P_{rec})}{\partial P_{rec}}$
is nondecreasing in $R$. By the similar procedure in proof of Lemma
\ref{lem:nondecreaseing}, we show this holds. \qed

\subsection{Proof of Proposition \ref{lem:When-,-i.e.II}}

We prove the proposition by induction. Suppose that the time horizon
consists of any $T$ time slots.

When $t=T$, $V_{T}(R)=U_{R}= RB$,
and the proposition is trivially true.

Now, we assume it also holds for $V_{t}(R)$ when $t=k+1,k+2,...,T.$
Let $\hat{R}$ be a system state such that $\hat{R}\geq R$. By the
hypothesis, we have $V_{k+1}(\hat{R})\geq V_{k+1}(R)$.
Let $\pi^{*}$ be the optimal policy. From the Bellman equation in
(\ref{eq:5647}), we have%
\begin{equation}
V_{k}(R)=\sum_{R'=0}^{\min\{M,N\}}P{}_{R,R'}^{\pi^{*}(R)}[U_{R'}+\beta V_{k+1}(R')],\forall R\in\mathcal{R}.\label{eq:111-1}
\end{equation}
By defining a new system state $-1$ such that
$U_{-1}+\beta V_{k+1}(-1)=0$,
we can rewrite the equation in (\ref{eq:111-1}) as\begin{eqnarray*}
V_{k}(R) & = & \sum_{R'=0}^{\min\{M,N\}}P{}_{R,R'}^{\pi^{*}(R)}\sum_{i=0}^{R'}\{[U_{i}+\beta V_{k+1}(i)]\\
 & & -[U_{i-1}+\beta V_{k+1}(i-1)]\}\\
 & = & \sum_{R'=0}^{\min\{M,N\}}\{[U_{R'}+\beta V_{k+1}(R')]\\
 & & - [U_{R'-1}+\beta V_{k+1}(R'-1)]\}\sum_{i=R'}^{\min\{M,N\}}P{}_{R,i}^{\pi^{*}(R)}.\end{eqnarray*}
By lemma $5$ in the Appendix, we have \[
\sum_{i=R'}^{\min\{M,N\}}P{}_{\hat{R},i}^{\pi^{*}(R)}\geq\sum_{i=R'}^{\min\{M,N\}}P{}_{R,i}^{\pi^{*}(R)},\forall R'\in\mathcal{R}.\]
Then\begin{eqnarray*}
V_{k}(R) & \le & \sum_{R'=0}^{\min\{M,N\}}\{[U_{R'}+\beta V_{k+1}(R')]\\
 &  & -[U_{R'-1}+\beta V_{k+1}(R'-1)]\}\sum_{i=R'}^{\min\{M,N\}}P{}_{\hat{R},i}^{\pi^{*}(R)}\\
 \end{eqnarray*}
 \begin{eqnarray*}
 & = & \sum_{R'=0}^{\min\{M,N\}}P{}_{\hat{R},R'}^{\pi^{*}(R)}[U_{R'}+\beta V_{k+1}(R')]\\
 & \le & \max_{P_{rec}\in\mathcal{P}}\sum_{R'\in\mathcal{R}}P{}_{\hat{R},R'}^{P_{rec}}[U_{R'}+\beta V_{t+1}(R')]\\
 & = & \sum_{R'=0}^{\min\{M,N\}}P{}_{\hat{R},R'}^{\pi^{*}(\hat{R})}[U_{R'}+\beta V_{k+1}(R')]\\
 & = & V_{k}(\hat{R}),\end{eqnarray*}
i.e., for $t=k$,
$V_{k}(\hat{R})\geq V_{k}(R)$
also holds. This completes the proof. \qed

\subsection{Proof of Theorem \ref{theorem1}}\label{proof_for_theorem1}

We first show that
under the reference distribution, the optimal policy is attainable.
\begin{lem}\label{thm33}
For the MRAS algorithm, the policy $\pi$ generated by the sequence
of reference distributions $\{g_{k}\}$ converges point-wisely to
the optimal spectrum access policy $\pi^{*}$ for the adaptive channel
recommendation MDP, i.e. \begin{eqnarray}
\lim_{k\rightarrow\infty}E_{g_{k}}[\pi(R)] & = & \pi(R)^{*},\forall R\in\mathcal{R},\label{eq:lemma11-1}\\
\lim_{k\rightarrow\infty}Var_{g_{k}}[\pi(R)] & = & 0,\forall R\in\mathcal{R}.\label{eq:lemma11-2}\end{eqnarray}
\end{lem}

\emph{proof:} The proof is developed on the basis of the results in \cite{key-5}.

First, from the MRAS algorithm, we have \[
\gamma_{k}\le\gamma_{k+1},\]
i.e. the sequence $\{\gamma_{k}\}$ is monotone. Since $0\le\gamma_{k}\le\Phi_{\pi^{*}}$
is bounded, there must exist a finite $K$ such that $\gamma_{k+1}=\gamma_{k},\forall k\geq K.$

When $\gamma_{K}=\Phi_{\pi^{*}}$, we have \[
\lim_{k\rightarrow\infty}E_{g_{k}}[e^{\Phi_{\pi}}I_{\{\Phi_{\pi}\geq\gamma_{k}\}}]=e^{\Phi_{\pi^{*}}}.\]
holds.

When $\gamma_{K}<\Phi_{\pi^{*}}$, from (\ref{eq:2354}), we know that\[
E_{g_{k}}[e^{\Phi_{\pi}}I_{\{\Phi_{\pi}\geq\gamma_{k}\}}]\geq E_{g_{k-1}}[e^{\Phi_{\pi}}I_{\{\Phi_{\pi}\geq\gamma_{k}\}}],\forall k\geq K.\]
That is, the sequence $\{E_{g_{k}}[e^{\Phi_{\pi}}I_{\{\Phi_{\pi}\geq\gamma_{k}\}}]\}$
is monotone and hence converges. We then show that the limit of this sequence must be $e^{\Phi_{\pi^{*}}}$
by contradiction.

Suppose that\[
\lim_{k\rightarrow\infty}E_{g_{k}}[e^{\Phi_{\pi}}I_{\{\Phi_{\pi}\geq\gamma_{k}\}}]=e^{\Phi_{*}}<e^{\Phi_{\pi^{*}}}.\]
Define the set\[
\Theta=\{\pi:\Phi_{\pi}\geq\max\{\gamma_{K},\ln\frac{e^{\Phi_{*}}+e^{\Phi_{\pi^{*}}}}{2}\}\}.\]
Since $\gamma_{K}<\Phi_{\pi^{*}}$, the set $\Theta$ is not empty
by the continuous property over the policy space of MDP \cite{key-4}.
Note that\[
g_{k}(\pi)=\prod_{i=1}^{k}\frac{e^{\Phi_{\pi}}I_{\{\Phi_{\pi}\geq\gamma_{i}\}}g_{k-1}(\pi)}{E_{g_{i}}[e^{\Phi_{\pi}}I_{\{\Phi_{\pi}\geq\gamma_{i}\}}]}g_{1}(\pi),\]
and \[
\lim_{k\rightarrow\infty}\frac{e^{\Phi_{\pi}}I_{\{\Phi_{\pi}\geq\gamma_{_{k}}\}}}{E_{g_{k}}[e^{\Phi_{\pi}}I_{\{\Phi_{\pi}\geq\gamma_{_{k}}\}}]}=\frac{e^{\Phi_{\pi}}I_{\{\Phi_{\pi}\geq\gamma_{_{K}}\}}}{e^{\Phi_{*}}}>1,\forall\pi\in\Theta,\]
we thus have\[
\lim_{k\rightarrow\infty}g_{k}(\pi)=\infty,\forall\pi\in\Theta.\]
By Fatou's lemma, we have\begin{eqnarray*}
 &  & \lim_{k\rightarrow\infty}\inf\int_{\pi\in\Omega}g_{k}(\pi)d\pi\\
 & = & 1\\
 & \geq & \lim_{k\rightarrow\infty}\inf\int_{\pi\in\Theta}g_{k}(\pi)d\pi\\
 & \geq & \int_{\pi\in\Theta}\lim_{k\rightarrow\infty}\inf g_{k}(\pi)d\pi\\
 & = & \infty,\end{eqnarray*}
which forms a contradiction. Hence, we have\[
\lim_{k\rightarrow\infty}E_{g_{k}}[e^{\Phi_{\pi}}I_{\{\Phi_{\pi}\geq\gamma_{k}\}}]=e^{\Phi_{\pi^{*}}}.\]

Since $e^{\Phi_{\pi}}I_{\{\Phi_{\pi}\geq\gamma\}}$ is a monotone
function of $\Phi_{\pi}$ and one-to-one map over the field $\{\pi:\Phi_{\pi}\geq\gamma\}$,
the result above implies that\begin{eqnarray}
\lim_{k\rightarrow\infty}E_{g_{k}}[\pi] & = & \pi^{*},\label{eq:lemma11-1}\\
\lim_{k\rightarrow\infty}Var_{g_{k}}[\pi] & = & \boldsymbol{0}.\label{eq:lemma11-2}\end{eqnarray} \qed

To complete the proof of the theorem, we next show that \[
E_{g_{k}}[\pi(R)]=E_{f(\pi,\boldsymbol{\mu},\boldsymbol{\sigma})}[\pi(R)],\forall R\in\mathcal{R},\]
\[E_{g_{k}}[\pi^{2}(R)]=E_{f(\pi,\boldsymbol{\mu},\boldsymbol{\sigma})}[\pi^{2}(R)],\forall R\in\mathcal{R}.\]
For the sake of simplicity, we first define a function \[
H(\boldsymbol{\mu},\boldsymbol{\sigma},\gamma_{k})\triangleq\int_{\pi\in\Omega}e^{(k-1)\Phi_{\pi}}I_{\{\Phi_{\pi}\geq\gamma_{k}\}}\ln f(\pi,\boldsymbol{\mu},\boldsymbol{\sigma})d\pi.\]
Since \begin{eqnarray*}
f(\pi,\boldsymbol{\mu},\boldsymbol{\sigma}) & = & \prod_{R=0}^{\min\{M,N\}}f(\pi(R),\mu_{R},\sigma_{R})\\
 & = & \prod_{R=0}^{\min\{M,N\}}\frac{1}{\sqrt{2pi\sigma_{R}^{2}}}e^{-\frac{(\pi(R)-\mu_{R})^{2}}{2\sigma_{R}^{2}}},\\
 & = & \prod_{R=0}^{\min\{M,N\}}e^{\frac{\mu_{R}\pi(R)}{\sigma_{R}}-\frac{\mu_{R}^{2}}{2\sigma_{R}}}\frac{1}{\sqrt{2pi\sigma_{R}^{2}}}e^{-\frac{\pi(R)^{2}}{2\sigma_{R}^{2}}}\\
 & = & \prod_{R=0}^{\min\{M,N\}}e^{\frac{\mu_{R}\pi(R)}{\sigma_{R}}-\frac{\mu_{R}^{2}}{2\sigma_{R}}}f(\pi(R),0,\sigma_{R})\\
 & = & \prod_{R=0}^{\min\{M,N\}}[e^{\frac{\mu_{R}\pi(R)}{\sigma_{R}}}f(\pi(R),0,\sigma_{R})\\
 &  & \cdot\int_{\pi(R)\in\mathcal{P}}\frac{\mu_{R}\pi(R)}{\sigma_{R}}f(\pi(R),0,\sigma_{R})d\pi(R)],\end{eqnarray*}
we then obtain\begin{eqnarray*}
 &  & H(\boldsymbol{\mu},\boldsymbol{\sigma},\gamma_{k})\\
 & = & \sum_{R=0}^{\min\{M,N\}}\int_{\pi\in\Omega}e^{(k-1)\Phi_{\pi}}I_{\{\Phi_{\pi}\geq\gamma_{k}\}}\frac{\mu_{R}\pi(R)}{\sigma_{R}}d\pi\\
 & & + \sum_{R=0}^{\min\{M,N\}}\int_{\pi\in\Omega}e^{(k-1)\Phi_{\pi}}I_{\{\Phi_{\pi}\geq\gamma_{k}\}}\ln f(\pi(R),0,\sigma_{R})d\pi\\
 & & - \sum_{R=0}^{\min\{M,N\}}\{\int_{\pi\in\Omega}e^{(k-1)\Phi_{\pi}}I_{\{\Phi_{\pi}\geq\gamma_{k}\}}\\
 &  & \cdot\ln[\int_{\pi(R)\in\mathcal{P}}\frac{\mu_{R}\pi(R)}{\sigma_{R}}f(\pi(R),0,\sigma_{R})d\pi(R)]d\pi\}.\end{eqnarray*}
Since the optimization problem in (\ref{eq:235-1}) is to solve\[
\max_{\boldsymbol{\mu},\boldsymbol{\sigma}}H(\boldsymbol{\mu},\boldsymbol{\sigma},\gamma_{k}),\]
the updated parameters ($\boldsymbol{\mu}_{k},\boldsymbol{\sigma}_{k}$)
thus maximizes $H(\boldsymbol{\mu},\boldsymbol{\sigma},\gamma_{k})$.
It means that \[
\nabla H(\boldsymbol{\mu}_{k},\boldsymbol{\sigma}_{k},\gamma_{k})=0.\]
That is \begin{eqnarray*}
 &  & \nabla H(\boldsymbol{\mu},\boldsymbol{\sigma},\gamma_{k})\\
 & = & \frac{\int_{\pi(R)\in\mathcal{P}}e^{\frac{\mu_{R}\pi(R)}{\sigma_{R}}}f(\pi(R),0,\sigma_{R})\frac{\pi(R)}{\sigma_{R}^{2}}d\pi(R)}{\int_{\pi(R)\in\mathcal{P}}e^{\frac{\mu_{R}\pi(R)}{\sigma_{R}}}f(\pi(R),0,\sigma_{R})d\pi(R)}\\
 &  & \cdot\int_{\pi\in\Omega}e^{(k-1)\Phi_{\pi}}I_{\{\Phi_{\pi}\geq\gamma_{k}\}}d\pi\\
 & & - \int_{\pi\in\Omega}e^{(k-1)\Phi_{\pi}}I_{\{\Phi_{\pi}\geq\gamma_{k}\}}\frac{\pi(R)}{\sigma_{R}^{2}}d\pi,\\
 & = & 0.\end{eqnarray*}
It follows that\begin{eqnarray*}
 &  & \frac{\int_{\pi\in\Omega}e^{(k-1)\Phi_{\pi}}I_{\{\Phi_{\pi}\geq\gamma_{k}\}}\pi(R)d\pi}{\int_{\pi\in\Omega}e^{(k-1)\Phi_{\pi}}I_{\{\Phi_{\pi}\geq\gamma_{k}\}}d\pi}\\
 & =\\
 &  & \frac{\int_{\pi(R)\in\mathcal{P}}e^{\frac{\mu_{R}\pi(R)}{\sigma_{R}}}f(\pi(R),0,\sigma_{R})\pi(R)d\pi(R)}{\int_{\pi(R)\in\mathcal{P}}e^{\frac{\mu_{R}\pi(R)}{\sigma_{R}}}f(\pi(R),0,\sigma_{R})d\pi(R)},\forall R\in\mathcal{R}.\end{eqnarray*}
By multiplying the same constant on the numerator and denominator
of the terms on both sides, we have\begin{eqnarray*}
 &  & \frac{\int_{\pi\in\Omega}\frac{e^{(k-1)\Phi_{\pi}}I_{\{\Phi_{\pi}\geq\gamma_{k}\}}g_{k-1}(\pi)}{E_{g_{k-1}}[e^{\Phi_{\pi}}I_{\{\Phi_{\pi}\geq\gamma\}}]}\pi(R)d\pi}{\int_{\pi\in\Omega}\frac{e^{(k-1)\Phi_{\pi}}I_{\{\Phi_{\pi}\geq\gamma_{k}\}}g_{k-1}(\pi)}{E_{g_{k-1}}[e^{\Phi_{\pi}}I_{\{\Phi_{\pi}\geq\gamma\}}]}d\pi}\\
 & =\\
 &  & \frac{\int_{\pi(R)\in\mathcal{P}}f(\pi(R),\mu_{R},\sigma_{R})\pi(R)d\pi(R)}{\int_{\pi(R)\in\mathcal{P}}f(\pi(R),\mu_{R},\sigma_{R})d\pi(R)},\forall R\in\mathcal{R},\end{eqnarray*}
Since \begin{eqnarray*}
 &  & \int_{\pi(R)\in\mathcal{P}}f(\pi(R),\mu_{R},\sigma_{R})d\pi(R)\\
 & = & \int_{\pi\in\Omega}\frac{e^{(k-1)\Phi_{\pi}}I_{\{\Phi_{\pi}\geq\gamma_{k}\}}g_{k-1}(\pi)}{E_{g_{k-1}}[e^{\Phi_{\pi}}I_{\{\Phi_{\pi}\geq\gamma\}}]}d\pi\\
 & = & 1,\end{eqnarray*}
we obtain\begin{eqnarray*}
 &  & \int_{\pi\in\Omega}\frac{e^{(k-1)\Phi_{\pi}}I_{\{\Phi_{\pi}\geq\gamma_{k}\}}g_{k-1}(\pi)}{E_{g_{k-1}}[e^{\Phi_{\pi}}I_{\{\Phi_{\pi}\geq\gamma\}}]}\pi(R)d\pi\\
 & = & \int_{\pi(R)\in\mathcal{P}}f(\pi(R),\mu_{R},\sigma_{R})\pi(R)d\pi(R),\forall R\in\mathcal{R},\end{eqnarray*}
i.e.\[
E_{g_{k}}[\pi(R)]=E_{f(\pi,\boldsymbol{\mu},\boldsymbol{\sigma})}[\pi(R)],\forall R\in\mathcal{R}.\]

Similarly, we can show that\[
E_{g_{k}}[\pi^{2}(R)]=E_{f(\pi,\boldsymbol{\mu},\boldsymbol{\sigma})}[\pi^{2}(R)],\forall R\in\mathcal{R}.\]

From (\ref{eq:lemma11-1}), it follows that\begin{eqnarray*}
\lim_{k\rightarrow\infty}E_{f(\pi,\boldsymbol{\mu}_{k},\boldsymbol{\sigma}_{k})}[\pi] & = & \lim_{k\rightarrow\infty}E_{g_{k}}[\pi]\\
 & = & \pi^{*}.\end{eqnarray*}
and,
\begin{eqnarray*}
 &   & \lim_{k\rightarrow\infty}Var_{f(\pi,\boldsymbol{\mu},\boldsymbol{\sigma})}[\pi(R)]\\
 & = & \lim_{k\rightarrow\infty}\{E_{f(\pi,\boldsymbol{\mu},\boldsymbol{\sigma})}[\pi^{2}(R)]-E_{f(\pi,\boldsymbol{\mu},\boldsymbol{\sigma})}[\pi(R)]^{2}\}\\
 & = & \lim_{k\rightarrow\infty}\{E_{g_{k}}[\pi^{2}(R)]-E_{g_{k}}[\pi(R)]^{2}\}\\
 & = & \lim_{k\rightarrow\infty}Var_{g_{k}}[\pi(R)]\\
 & = & 0.\end{eqnarray*} \qed
\bibliographystyle{ieeetran}
\bibliography{AdaptiveRec}


\end{document}